%% file: Main.tex
\documentclass[11pt,draftcls,onecolumn]{IEEEtran}
\IEEEoverridecommandlockouts

\makeatletter
\NewDocumentCommand{\seteqnum}{o}{%
  \IfValueTF{#1}
    {\textup{\tagform@{#1}}}% <label> supplied
    {\incr@eqnum \print@eqnum}% <label> not supplied
}
\NewCommandCopy{\ltxlabel}{\ltx@label}
\makeatother
\usepackage[cmex10]{amsmath} 
\usepackage{amsfonts,amsthm}
\usepackage{bm}
\usepackage{stmaryrd}
\usepackage{acronym}
\usepackage{amssymb}
\usepackage{dsfont}
\usepackage{subcaption}
\usepackage{graphicx}
\usepackage{paralist}
\usepackage{mdframed}
\usepackage{url}
\usepackage{cite}
\usepackage{hyperref}
\usepackage{flushend}
\usepackage{enumerate}
\usepackage{pdflscape}
\usepackage[dvipsnames]{xcolor}
\usepackage{mathrsfs}
\usepackage{lipsum}
\usepackage{diagbox}
\usepackage{soul}
\usepackage{kantlipsum}
\usepackage{braket}
\usepackage{physics}
\usepackage{stackengine}

\usepackage{color}
\usepackage{framed}
\usepackage{float}

\makeatletter
\def\blfootnote{\xdef\@thefnmark{}\@footnotetext}
\makeatother

\newcommand{\indi}[1]{\ensuremath{\mathds{1}}}

\allowdisplaybreaks
\allowbreak

\input{CommandsAndMacros.tex}
\input{Acronyms.tex}

\begin{document}

\title{From Classical to Quantum Channels: Achieving Positive Covert Rates}

\author{
\IEEEauthorblockN{Hassan ZivariFard,  Xiaodong Wang, and Alexei Ashikhmin}\\
\thanks{H. ZivariFard and X. Wang are with the Department of Electrical Engineering, Columbia University, New York, NY 10027. A. Ashikhmin is with Nokia Bell Labs, Murray Hill, NJ 07974 USA. The work of H. ZivariFard and X. Wang is supported in part by the U.S. Office of Naval Research (ONR) under grant N000142412212. E-mails: \{hz2863, xw2008\}@columbia.edu and alexei.ashikhmin@nokia-bell-labs.com. Part of this work will be presented at the 2026 IEEE Information Theory Workshop~(ITW).
}
}
\maketitle
\date{}

\begin{abstract}
\label{sec:Abstract}
In this paper, we study the conditions under which covert communication at positive rates is feasible over discrete memoryless classical, classical-quantum, and quantum channels. 
For classical point-to-point channels, we show that if the dimension of the channel input probability simplex exceeds that of the channel output probability simplex, equivalently, if the input alphabet has larger cardinality than the output alphabet, then positive covert rates are achievable for certain classes of \acp{DMC}. We further show that allowing the innocent symbol (i.e., the symbol transmitted in the no-communication mode) to be chosen appropriately can enlarge the class of \acp{DMC} for which positive covert rates are achievable. 
We also study covert communication over classical \acp{MAC}, where the additional transmitter effectively enlarges the set of available channel input pairs, and we show that the conditions required to satisfy the covertness constraint are less restrictive for \acp{MAC} than for point-to-point channels. 
We extend these results to classical-quantum \acp{DMC} by showing that if the dimension of the input probability simplex exceeds the affine dimension of the set of output states that can be induced at the channel output, then positive covert rates are achievable for certain classes of classical-quantum \acp{DMC}. 
Finally, for quantum channels, we show that if the innocent state (i.e., the state transmitted in the no-communication mode) is mixed, then the covertness constraint can always be satisfied by a non-trivial input ensemble. Consequently, positive covert rates are achievable whenever the legitimate receiver can distinguish at least two states in a suitable such ensemble.
\end{abstract}

\section{Introduction}
\label{sec:Intro}
Covert communication is an emerging research area with applications such as watermarking, steganography, and data embedding. However, it is well known that for classical point-to-point \acp{DMC}, the number of bits that can be transmitted covertly scales only on the order of the square root of the total number of channel uses, a phenomenon known as the square-root law \cite{Bash13,Bloch16,Wang16}. It has also been shown that the square-root law extends to classical-quantum channels \cite{WangCQ,Bullock25} and quantum channels \cite{Bash_15}. This limitation can be restrictive in many applications, as such scaling may not provide sufficient throughput. 
Positive-rate covert communication has been shown to be achievable in several settings through mechanisms that either modify the effective no-communication distribution or enlarge the set of distributions that can be induced during communication, as summarized below:
\begin{itemize}
    \item The warden has uncertainty about the channel statistics \cite{Deniable_ITW14,Goeckel16,He17,Lee15}: The warden's uncertainty about the channel statistics causes the no-communication output distribution to become a mixture over the possible channel realizations. This averaged no-communication distribution may lie in the convex hull of the output distributions induced during communication, thereby enabling positive covert communication rates. We note that \cite{Goeckel16,He17,Lee15} assume the availability of an unlimited-rate secret key shared between the transmitter and the receiver, whereas \cite{Deniable_ITW14} assumes that the warden's channel is degraded \ac{wrt} the legitimate receiver's channel.

    \item A friendly jammer is present \cite{Sobers17,Shahzad18,Shmuel19,ISIT21,ISIT22,MyDissertation}: the randomness introduced by the friendly jammer creates uncertainty about the statistics of the warden's channel output in both the no-communication and communication modes. In particular, while the no-communication output distribution is subject to the same type of uncertainty described above, the output distribution induced during communication is also randomized by the jammer. This additional uncertainty enables the transmitter to surpass the square-root law and achieve positive covert communication rates. We note that \cite{Sobers17,Shahzad18,Shmuel19} assume the availability of an unlimited-rate secret key shared between the transmitter and the receiver.

   \item Channel state information is available at the transmitter \cite{LeeWang18,Keyless22,Action_Covert,Quantum_Covert_CSI}: The availability of channel state information enlarges the set of output distributions that can be induced at the warden's channel output, thereby enabling positive covert communication rates.

    \item A cooperative user assists the covert communication \cite{ISIT22,ExtendedPaperMAC}: Similarly, the presence of a cooperative user enlarges the set of output distributions that can be induced at the warden's channel output. Consequently, positive covert communication rates may become achievable.

\end{itemize}

\subsection{Our Contributions}
For point-to-point \acp{DMC} without access to any auxiliary resources, such as those listed above, achieving covert communication at positive rates is generally considered impossible due to the square-root law. An exception arises when the innocent symbol, denoted by $x_0$ and transmitted in the no-communication mode, is \emph{redundant}, in the sense that the output distribution it induces at the warden can be replicated by a mixture of the distributions induced by the remaining input symbols \cite{Wang16}. However, to the best of our knowledge, existing work has largely treated this as an exceptional condition and provided almost no systematic characterization of channels that satisfy it. 
In this paper, we investigate covert communication over point-to-point and multiple-access classical, classical-quantum, and quantum channels. A central theme of this work is that covert communication can be understood through the geometry of the channel. Rather than viewing the covertness constraint purely as an information-theoretic condition, we interpret it as a question about whether the innocent output distribution belongs to the convex hull generated by the remaining channel input symbols. This geometric perspective unifies the classical, classical-quantum, and quantum settings and provides simple sufficient conditions for positive-rate covert communication. This paper shows that such channels naturally arise whenever the input alphabet is sufficiently rich relative to the output alphabet. 

For classical point-to-point DMCs, we establish a geometric characterization of the conditions under which the covertness constraint can be satisfied with a non-trivial input distribution. In particular, we show that when the channel input alphabet has larger cardinality than the channel output alphabet, equivalently, when the dimension of the input probability simplex exceeds that of the output probability simplex, positive covert rates are achievable for certain classes of DMCs. Moreover, increasing the cardinality of the input alphabet relative to that of the output alphabet progressively relaxes the conditions required to satisfy the covertness constraint. We further show that appropriately choosing the innocent symbol can enlarge the class of DMCs for which positive covert rates are achievable. In particular, for a binary-output DMC with more than two input symbols, an appropriate choice of the innocent symbol always permits a non-trivial covert input distribution.
As a natural mechanism for enlarging the set of available channel input pairs, we next study covert communication over classical \acp{MAC}. We extend the geometric characterization of the covertness constraint to the multi-user setting and derive conditions for the existence of non-trivial covert input distributions under both independent and joint channel inputs. We show that allowing joint channel inputs can relax the conditions required to satisfy the covertness constraint and can create covert communication opportunities that are unavailable under independent inputs. Thus, the presence of an additional transmitter can itself serve as a resource for achieving positive-rate covert communication.

We then extend the geometric characterization to classical-quantum channels. In particular, for classical-quantum DMCs, we replace the output probability simplex by the affine space generated by the channel output states. We show that if the dimension of the input probability simplex exceeds the affine dimension of the set of output states that can be induced at the channel output, then positive covert rates are achievable for certain classes of classical-quantum DMCs. Therefore, the analysis of the classical point-to-point and \ac{MAC} settings extends naturally to their classical-quantum counterparts. In each case, the condition under which the covertness constraint $\rho_Z=\sigma_0$ can be satisfied coincides with the corresponding classical condition $p_Z=q_0$.

Finally, the fully quantum setting exhibits fundamentally different behavior. We show that whenever the innocent input state is mixed, the covertness constraint can always be satisfied by a non-trivial input ensemble. This reveals a qualitative distinction between classical and quantum covert communication. We also study several known quantum channels for which positive covert rates are achievable even when the innocent input state is pure. 

For classical channels, we characterize positive-rate covert capacity, and for quantum and classical-quantum channels, we derive a general achievable covert rate. Our achievability scheme is based on random coding, pinching methods \cite{QIT_Hayashi}, and channel resolvability. Our results show that allowing stochastic encoding can enlarge the covert capacity relative to deterministic encoding.

\subsection{State of the Art}
Of particular relevance to this paper, covert communication over point-to-point Gaussian channels is studied in \cite{Bash13}, where it is shown that the covert capacity obeys the square-root law when the legitimate terminals share a secret key of unbounded rate. Covert communication over classical point-to-point \acp{DMC} is investigated in \cite{Bloch16,Wang16}, where it is likewise shown that the covert capacity follows the square-root law. 
Moreover, \cite{Wang16} characterizes the positive-rate covert capacity of classical \acp{DMC} in the special case where the legitimate receiver and the warden observe the same channel output and the legitimate terminals share a secret key of unbounded rate. In addition, covert communication over a classical $K$-user \ac{MAC} is investigated in \cite{Cover_K_User_MAC}, where the authors show that the covert capacity obeys the square-root law. Finally, the trade-off between the achievable covert rate and the secret-key rate is studied in \cite{Bounhar26}. Covert communication with positive rates over channels with action-dependent states is studied in \cite{Action_Covert}, where the authors show that the availability of channel state information can enable positive covert rates by enlarging the set of output distributions that can be induced at the channel output.

\subsection{Notation} 
We denote the set of real numbers by $\bbR$. Random variables are denoted by uppercase letters and their realizations by lowercase letters. Sets are denoted by calligraphic letters, the cardinality of a set is denoted by $\card{\cdot}$, and the convex hull of a set is denoted by $\mathrm{conv}\br{\cdot}$. Superscripts denote the dimension of a vector, e.g., $x^n\triangleq(x_1,x_2,\cdots,x_n)$. 
For a countable set $\calX$, the relative entropy between two distributions $p_X$ and $q_X$ is defined as $\bbD(p_X \lVert q_X) \triangleq \sum_{x \in \calX} p_X(x) \log \frac{p_X(x)}{q_X(x)}$ and the total variation between $p_X$ and $q_X$ is defined as $\lVert p_X-q_X\lVert_1=\frac{1}{2}\sum_x |p_X(x)-q_X(x)|$. 
All logarithms are taken to base $2$. The $n$-fold product distribution is denoted by $p_X^\on(x^n) \triangleq \prod_{i=1}^n p_X(x_i)$. 
We denote the set of positive semidefinite operators on a finite-dimensional Hilbert space $\calH$ by $\calP(\calH)$, and the set of quantum states by $\calD(\calH)\triangleq\br{\sigma\in\calP(\calH):\tra[\sigma]=1}$. 
The identity operator on $\calH$ is denoted by $\dsI$. For $\rho, \sigma \in \calD(\calH)$, the trace norm of the operator $\rho$ is defined as $\lVert\rho\rVert_1\triangleq\tra\sbr{\sqrt{\rho^\dagger\rho}}$, and the trace distance between $\rho$ and $\sigma$ is defined as $\lVert \rho - \sigma \rVert_1$.

\subsection{Paper Organization}
The remainder of the paper is organized as follows. Section~\ref{sec:Problem_Defini} introduces the system model. Section~\ref{sec:Classical} studies classical \acp{DMC} and develops the geometric characterization of positive covert rates. 
Section~\ref{sec:MAC} considers classical discrete memoryless \acp{MAC} and \acp{MAC} with degraded message sets.
Section~\ref{sec:Quantum} extends these results to classical-quantum and quantum channels. Section~\ref{sec:Conclusions} concludes the paper.

\section{Problem Statement}
\label{sec:Problem_Defini}
\begin{figure}
\centering
\vspace{-0.2cm}
\includegraphics[height=1.5in]{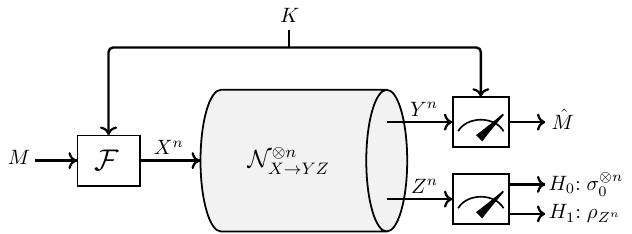}
\caption{Covert Communication over a point-to-point channel}
\label{fig:System_Model}
\end{figure}
\subsection{Quantum Channels}
\label{sec:Defi_Quantum}
Consider a quantum channel $\calN_{X\to YZ}$, which is a completely positive trace-preserving linear map that takes an operator $\rho_X\in\calD(\calH_X)$ to an operator $\rho_{YZ}\in\calD(\calH_Y\otimes\calH_Z)$. The code is formally defined as follows.
\begin{definition}
    \label{defi:Code}
    A $(2^{nR},2^{nR_K},n)$ code for the quantum channel $\calN_{X\to YZ}$ consists of the following:
    \begin{itemize}
        \item a message set $\calM\triangleq\sbr{1\!:\!2^{nR}}$ and a secret key set $\calK\triangleq\sbr{1\!:\!2^{nR_K}}$;
        \item a stochastic encoding map $\calF: \calM \times \calK \to \calD(\calH_{X^n})$ that maps each message $m \in \calM$ and secret key $k \in \calK$ to a channel input state $\rho^{(m,k)}_{X^n} \in \calD(\calH_{X^n})$;
        \item for each $k\in\calK$ a collection of \acp{POVM} $\br{\Lambda_{m}^{(k)}}_{m \in \calM}$, where $\sum_m\Lambda_{m}^{(k)}=\dsI$, acting on $\calH_{Y^n}$, which, given the secret key $k$, induces a decoding rule that maps the channel output state $\rho_{Y^n}^{(m,k)} \in \calD(\calH_{Y^n})$ to an estimate $\hat{m} \in \calM$.
    \end{itemize}
\end{definition}The code is public knowledge, and the objective is to design a code that is both reliable and covert. From Definition~\ref{defi:Code}, for each $(m,k)\in\calM\times\calK$, the channel input and the output of the legitimate receiver~are
\begin{align}
    \rho_{X^n}^{(m,k)}=\calF(m,k),\quad \rho_{Y^n}^{(m,k)}=\tra_{Z^n}\calN_{X\to YZ}^\on\pr{\rho_{X^n}^{(m,k)}}.
    \label{eq:Input_Output}
\end{align}Therefore, from Definition~\ref{defi:Code} and \eqref{eq:Input_Output} we have
\begin{align}
    p_e^{(n)}\triangleq\bbP\br{\hat{M}\ne M}=\frac{1}{\card{\calM}\card{\calK}}\sum_{(m,k)\in\calM\times\calK}\tra\sbr{\pr{\dsI-\Lambda_{m}^{(k)}}\rho_{Y^n}^{(m,k)}}.\label{eq:P_Error}
\end{align}A sequence of codes is reliable if
\begin{align}
    \lim_{n\to\infty}p_e^{(n)}=0.\label{eq:limit_Perror}
\end{align}
When communication is happening, the warden observes the state
\begin{align}
    \rho_{Z^n}\triangleq\frac{1}{\card{\calM}\card{\calK}}\sum_{m\in\calM}\sum_{k\in\calK}\tra_{Y^n}\calN_{X\to YZ}^\on\pr{\rho_{X^n}^{(m,k)}}.\label{eq:comm_Warden}
\end{align}
When communication does not take place, the transmitter always transmits the innocent state $\rho_0\in\calD\pr{\calH_X}$ over the channel; as a result, the warden's output in the no-communication mode~is
\begin{align}
    \sigma_0^\on\triangleq\tra_{Y^n}\calN_{X\to YZ}^\on\pr{\rho_0^\on}.\label{eq:nocomm_Warden}
\end{align}The covertness constraint is defined as
\begin{align}
    \lim_{n\to\infty}\lVert\rho_{Z^n}-\sigma_0^\on\rVert_1=0.\label{eq:Covertness_Const}
\end{align}
\begin{definition}
\label{defi:Code_Capa}
    A rate pair $(R, R_K)$ is said to be achievable for the quantum channel $\calN_{X \to YZ}$ if there exists a sequence of $(2^{nR}, 2^{nR_K}, n)$ codes such that \eqref{eq:limit_Perror} and \eqref{eq:Covertness_Const} are satisfied. The covert capacity region, denoted by $\calC_{\mathrm{C\text{-}Q}}$, is defined as the closure of the set of all achievable rate pairs.
\end{definition}
\subsection{Classical-Quantum and Classical Channels}
An important special case of the problem defined in Section~\ref{sec:Defi_Quantum} corresponds to covert communication over classical-quantum channels. In this case, the channel is denoted by $\pr{\calX,\br{\rho_{YZ}^{(x)}}_{x\in\calX}}$, where $\calX$ is the channel input alphabet and $\rho_{YZ}^{(x)}$ is the density operator associated with input symbol $x \in \calX$. 
The encoder $\calF$ takes as input a message $m \in \calM$ and a secret key $k \in \calK$, and outputs a channel input sequence $X^n(m,k) \in \calX^n$. The decoder is defined analogously to Definition~\ref{defi:Code}.

For classical-quantum channels, the innocent input state $\rho_0$ is replaced by an innocent input symbol $x_0$, but the covertness constraint remains the same as in \eqref{eq:Covertness_Const}. Under the above problem setup, Definitions~\ref{defi:Code} and~\ref{defi:Code_Capa} extend naturally to the setting of classical-quantum channels. In this context, the classical-quantum covert capacity is denoted by $\calC_{\mathrm{C\text{-}CQ}}$.

Similarly, the problem of covert communication over classical channels can be formulated by drawing on the definitions in Section~\ref{sec:Defi_Quantum} for quantum channels and those given above for classical-quantum channels. In this setting, the channel is denoted by $\pr{\calX, W_{YZ|X}, \calY, \calZ}$, where $\calX$ is the input alphabet, $W_{YZ|X}$ is the channel law, and $\calY$ and $\calZ$ are the output alphabets at the legitimate receiver and the warden, respectively. The encoder is defined analogously to that of the classical-quantum case, while the decoder $\Lambda_m^{(k)}$ maps the channel output $Y^n$ to an estimate of the message, denoted by $\hat{m}$.

For classical channels, the innocent input state $\rho_0$ is replaced by an innocent input symbol $x_0$, and $\sigma_0^{\otimes n}$ in~\eqref{eq:nocomm_Warden} corresponds to the distribution induced at the warden's output in the absence of communication. Moreover, in the covertness constraint~\eqref{eq:Covertness_Const}, the trace distance is replaced by the total variation distance. Under this formulation, Definitions~\ref{defi:Code} and~\ref{defi:Code_Capa} extend naturally to classical channels. In this context, the classical covert capacity is denoted by $\calC_{\mathrm{C\text{-}C}}$.

\section{Classical Point-to-Point Channels}
\label{sec:Classical}
In this section, we study covert communication at positive rates over classical point-to-point \acp{DMC}. We begin by establishing the covert capacity for general DMCs, then focus on binary‑output channels, which exhibit several useful structural properties. Finally, we extend these results to channels with arbitrary output alphabets.
\begin{theorem}[Covert Capacity of Classical Channels]
    \label{thm:Capacity_Classical_Sto}
\begin{subequations}\label{eq:Capacity_Classical_All_Sto}
The covert capacity of the classical \ac{DMC} $W_{YZ\lvert X}$,~is
\begin{align}
\calC_{\mathrm{C\text{-}C}} =\bigcup_{p_{UXYZ}\in\calB}\left.\begin{cases}(R,R_K):\\
  R\le\bbI(U;Y),\\
  R_K\ge\sbr{\bbI(U;Z)-\bbI(U;Y)}^+,
\end{cases}\hspace{-3mm}\right\},
\label{eq:Capacity_Classical_Sto}
\end{align}
where
\begin{align}
  \calB \triangleq \left.\begin{cases}p_{UXYZ}:\\
p_{UXYZ}=p_Up_{X|U}W_{YZ\lvert X},\\
p_Z=q_0\triangleq W_{Z|X=x_0},\\
\end{cases}\right\}.\label{eq:Capacity_Classical_S_Sto}
\end{align}
\end{subequations}
\end{theorem}
The achievability proof of Theorem~\ref{thm:Capacity_Classical_Sto} follows from that of Theorem~\ref{thm:Achievable_Quantum} in Appendix~\ref{proof:thm:Achievable_Quantum} and the converse proof is provided in Appendix~\ref{proof:thm:capacity}.

\begin{remark}[Deterministic Encoder vs.\ Stochastic Encoder]
\label{remark:Sto_Deter}
The covert capacity in Theorem~\ref{thm:Capacity_Classical_Sto} is based on stochastic encoding, which is implemented via channel prefixing~\cite{BCC:IT78}. Deterministic encoding is recovered as the special case $U=X$. Consequently, the covert capacity achieved with stochastic encoding is always at least as large as that achieved with deterministic encoding, and can be strictly larger. This gain is due to the additional design freedom provided by the auxiliary random variable $U$, which can reduce the term $\bbI(X;Z)-\bbI(X;Y)$, in a manner analogous to the classical broadcast channel with confidential messages~\cite{BCC:IT78}.
\end{remark}
\begin{remark}
    Setting $Y=Z$, the capacity region in Theorem~\ref{thm:Capacity_Classical_Sto} recovers the capacity region in \cite[Proposition~1]{Wang16}.
\end{remark}
In the sequel, for simplicity, we restrict our attention to the deterministic encoding scheme, which is achieved by setting $U=X$ in Theorem~\ref{thm:Capacity_Classical_Sto}. Nevertheless, the results can be extended to the stochastic encoding scheme presented in Theorem~\ref{thm:Capacity_Classical_Sto}.

It is well known that the covert capacity of point-to-point channels obeys the square-root law \cite{Bash13,Bloch16,Wang16}. This behavior arises because, for many channels, the covertness constraint $p_Z = q_0$ in \eqref{eq:Capacity_Classical_S_Sto} of Theorem~\ref{thm:Capacity_Classical_Sto} cannot be satisfied; a notable example is the \ac{AWGN} channel \cite{Bash13}. Consider a classical communication channel $\pr{\calX, W_{YZ|X}, \calY, \calZ}$, where $\calX$ denotes the channel input alphabet and $\calY$ and $\calZ$ are the output alphabets observed by the legitimate receiver and the warden, respectively. The covertness constraint $p_Z = q_0$ can be satisfied, and hence a positive covert communication rate is achievable, if the output distribution observed by the warden under the innocent input $X=x_0$, denoted by $W_{Z|X=x_0}$, can be expressed as a convex combination of the output distributions corresponding to the remaining channel inputs \cite{Wang16}, i.e.,
\begin{align}
    W_{Z|X=x_0}\in\mathrm{conv}\br{W_{Z|X=x}:x\in\calX\backslash\br{x_0}}.\nonumber
\end{align} In what follows, we show that for point-to-point \acp{DMC}, the constraint $p_Z = q_0$ can be satisfied under crtain conditions, and consequently, a positive covert communication rate becomes achievable. In particular, when the channel input alphabet is larger than the output alphabet, the constraint $p_Z = q_0$ can be satisfied for certain classes of \acp{DMC}, thereby enabling strictly positive covert communication rates. Moreover, increasing the input alphabet relative to the output alphabet progressively enlarges the class of \acp{DMC} for which the covertness constraint can be satisfied.

We begin by studying \acp{DMC} with a binary-output alphabet, which exhibits interesting properties, and then extend the analysis to channels with arbitrary output-alphabet sizes. 
\subsection{Discrete-Input and Binary-Output Warden Channels}
\label{sec:DIS_Inp_BO}
Consider an $m$-input binary-output \ac{DMC} with parameters $\epsilon_i\in\sbr{0,\!1}$, for $i\in\sbr{0\!:\!m-1}$, as illustrated in Fig.~\ref{fig:GIBO}, and assume that the innocent symbol is $x_0=j$, for some $j\in\sbr{0\!:\!m-1}$. 
\begin{theorem}
\label{thm:General_DMC}
    For the binary-output \ac{DMC} depicted in Fig.~\ref{fig:GIBO}, assuming that $x_0=j$, for some $j\in\sbr{0\!:\!m-1}$, there exists a non-degenerate input distribution satisfying the covertness constraint $p_Z=q_0$, if and only if
\begin{align}
    \min\limits_{\substack{i\in[0:m-1]\\i\ne j}}\epsilon_i\le\epsilon_j\le\max\limits_{\substack{i\in[0:m-1]\\i\ne j}}\epsilon_i.\label{eq:Coverness_GIBO_Final}
\end{align}
Furthermore, when $m\geq 3$ and the innocent symbol $x_0$ can be freely chosen, the covertness constraint $p_Z=q_0$ can always be satisfied.
\end{theorem}
\begin{proof}
Let $p_X(i)=p_i$ for $i\in\sbr{0\!:\!m-1}$, where $p_i\in[0,\!1]$ and $\sum_{i=0}^{m-1}p_i=1$. Then,
\begin{align}
p_Z(0)&=\sum_{i=0}^{m-1}p_i\epsilon_i.\label{eq:Possible_Set_Dist_Output_GIBO}
\end{align}Therefore, by varying the parameter $p_i$, $i\in\sbr{0\!:\!m-1}$, we can characterize the set of all distributions that can be induced at the channel output, as illustrated in Fig.~\ref{fig:GIBO_Dist}. Equivalently, the set of all probability distributions that can be induced at the channel 
output is given by
\begin{equation}
    \Delta_{\mathrm{GIBO}}\triangleq\mathrm{conv}\left\{(1-\epsilon_i,\epsilon_i)\!:i\in[0:m-1]\right\}.\label{eq:Conv_GIBO}
\end{equation}We also have, $q_0(0)=\epsilon_j$. For a non-degenerate input distribution, i.e., $\max_i p_i<1$, the covertness constraint $p_Z=q_0$ is feasible if and only if
\begin{figure*}[t!]
    \centering
    \hspace{-5mm}\begin{subfigure}[t]{0.45\textwidth}
        \centering
        \includegraphics[height=2.0in]{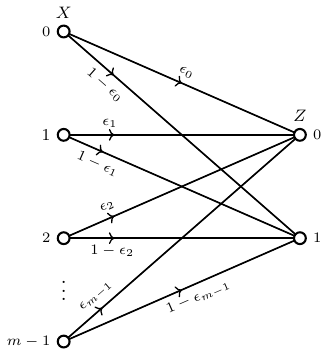}
        \caption{General Binary-Output \ac{DMC}}
        \label{fig:GIBO}
    \end{subfigure}%
    ~ 
    \hspace{-5mm}\begin{subfigure}[t]{0.55\textwidth}
        \centering
        \includegraphics[height=2.2in]{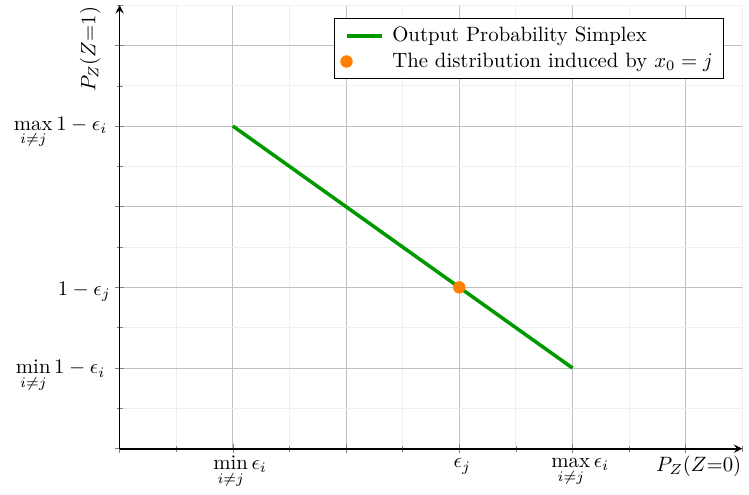}
        \caption{the set of all possible induced output distributions}
        \label{fig:GIBO_Dist}
    \end{subfigure}
    \caption{A general binary-output \ac{DMC} with $W_{Z|X}(0|i)=\epsilon_i$ and $W_{Z|X}(1|i)=1-\epsilon_i$, for $i\in[0\!:\!m-1]$, and the set of all output distributions induced by varying the channel input distribution}
    \label{fig:GIBO_Example}
    \vspace{-0.7cm}
\end{figure*}
\begin{align}
     p_Z(0)=q_0(0)&\Leftrightarrow\sum_{i=0}^{m-1}p_i\epsilon_i=\epsilon_j,\nonumber\\
    &\Leftrightarrow\sum_{\substack{i=0\\ i\ne j}}^{m-1}p_i\epsilon_i = (1-p_j)\epsilon_j\nonumber\\
    &\mathop{\Leftrightarrow}\limits^{(a)}\sum_{\substack{i=0\\ i\ne j}}^{m-1}\frac{p_i\epsilon_i}{1-p_j} = \epsilon_j\nonumber\\
    &\mathop{\Leftrightarrow}\limits^{(b)}\sum_{\substack{i=0\\ i\ne j}}^{m-1}\lambda_i\epsilon_i = \epsilon_j,\label{eq:Coverness_GIBO}
\end{align}where $(a)$ follows since $p_j<1$, and $(b)$ follows by defining $\lambda_i\triangleq\frac{p_i}{1-p_j}$, for $i\in[0\!:\!m-1]$ and $i\ne j$, 
which satisfy $\lambda_i\ge0$, and $\sum\limits_{\substack{i=0\\ i\ne j}}^{m-1}\lambda_i=1$. Note that the \ac{LHS} of \eqref{eq:Coverness_GIBO} is by definition a convex combination of $\epsilon_i$, for $i\in[0\!:\!m-1]$ and $i\ne j$, which is the interval
\begin{align}
    \sbr{\min\limits_{\substack{i\in[0:m-1]\\i\ne j}}\epsilon_i,\max\limits_{\substack{i\in[0:m-1]\\i\ne j}}\epsilon_i}.\nonumber
\end{align}
Therefore, the necessary and sufficient condition to satisfy the covertness constraint $p_Z=q_0$ is~\eqref{eq:Coverness_GIBO_Final}.

To see that, for any $m\geq 3$, there always exists an innocent symbol for which the covertness constraint in \eqref{eq:Coverness_GIBO_Final} is satisfied, consider the following. Let $\epsilon_j$ be a median of the $m$ values $\{\epsilon_i\}_{i=0}^{m-1}$. Since $m\geq3$, after excluding $\epsilon_j$, there exists at least one $\epsilon_i$ satisfying $\epsilon_i\leq\epsilon_j$ and at least one $\epsilon_i$ satisfying $\epsilon_i\geq\epsilon_j$. Consequently, \eqref{eq:Coverness_GIBO_Final} and hence the covertness constraint can be satisfied by choosing $x_0=j$.
\end{proof}
\begin{remark}[Intuition]
\label{rem:innocent_symbol_bin}
    Intuitively, when the channel output is binary and the input alphabet is strictly larger than the output alphabet, the output distributions induced by the input symbols cannot all be extreme points of the output probability simplex. Consequently, at least one induced output distribution lies in the convex hull of the others, which enables non-trivial covert communication.

    As an illustrative example, let $m=3$, with $\epsilon_0 = 0.1$, $\epsilon_1 = 0.2$, and $\epsilon_2 = 0.3$. In this case, the constraint in~\eqref{eq:Coverness_GIBO_Final} is not satisfied when the innocent symbol $x_0=0$. However, if we instead select $x_0 = 1$ as the innocent symbol, the constraint in~\eqref{eq:Coverness_GIBO_Final} will be satisfied.
\end{remark}
When the innocent symbol $x_0$ is fixed, the condition in \eqref{eq:Coverness_GIBO_Final} can become less restrictive as the cardinality $m$ of the channel input alphabet increases, i.e., as additional input symbols are added. To illustrate this, we first consider binary-input, binary-output \acp{DMC} and then successively increase the cardinality of the input alphabet by considering ternary-input, binary-output and quaternary-input, binary-output \acp{DMC}. These examples demonstrate that enlarging the channel input alphabet can only enlarge the set of output distributions that can be induced at the warden and therefore can only relax the covertness constraint. Finally, we consider the extreme case in which the channel input alphabet is continuous while the channel output alphabet remains binary, and show that the covertness constraint can always be satisfied in this case.
\begin{figure*}[t!]
    \centering
    \hspace{-5mm}\begin{subfigure}[t]{0.45\textwidth}
        \centering
        \includegraphics[height=2.0in]{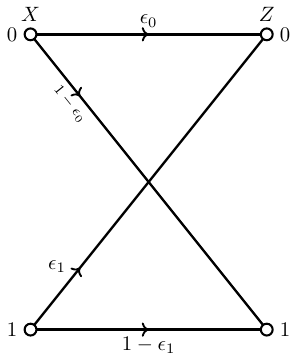}
        \caption{Binary-input, binary-output \ac{DMC}}
        \label{fig:BIBO}
    \end{subfigure}
    ~ 
    \hspace{-5mm}\begin{subfigure}[t]{0.55\textwidth}
        \centering
        \includegraphics[height=2.2in]{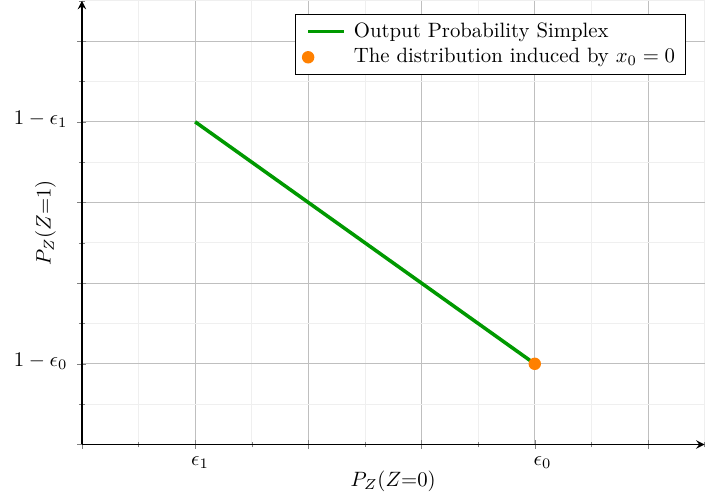}
        \caption{the set of all output distributions induced by the channel}
        \label{fig:BIBO_Dist}
    \end{subfigure}
    \caption{A general binary-input, binary-output \ac{DMC} and the set of all output distributions induced by varying the channel input distribution}
    \label{fig:BIBO_Example}
    \vspace{-0.7cm}
\end{figure*}

\subsubsection{Binary-Input, Binary-Output Warden \texorpdfstring{\acp{DMC}}{DMCs}}
When the channel input alphabet of the \ac{DMC} illustrated in Fig.~\ref{fig:GIBO} is $\calX=\br{0,1}$ and the innocent input symbol is $x_0=0$, as illustrated in Fig.~\ref{fig:BIBO}, the covertness constraint in \eqref{eq:Coverness_GIBO_Final} reduces to $\epsilon_1\le\epsilon_0\le\epsilon_1$, which can be satisfied only if
\begin{align}
\epsilon_0=\epsilon_1.\label{eq:Cov_BIBO_Final}
\end{align}Thus, the covertness constraint imposes a highly restrictive condition on the channel. Under this condition, the two input symbols $0$ and $1$ induce the same output distribution, $\bigl(p_Z(0), p_Z(1)\bigr) = (\epsilon_0,1-\epsilon_0)$, and hence the warden's output distribution is identical in the communication and no-communication modes. Consequently, communication is always covert, regardless of the input distribution, and the covert capacity coincides with the Shannon capacity of the legitimate receiver's channel. As illustrated in Fig.~\ref{fig:BIBO_Dist}, the set of all probability distributions that can be induced at the channel output in \eqref{eq:Conv_GIBO} therefore reduces to
\begin{equation}
    \Delta_{\mathrm{BIBO}}
    \triangleq
    \mathrm{conv}\left\{
        (\epsilon_0,1-\epsilon_0),
        (\epsilon_1,1-\epsilon_1)
    \right\}.
    \label{eq:Conv_BIBO}
\end{equation}

\begin{figure*}[t!]
    \centering
    \hspace{-5mm}\begin{subfigure}[t]{0.45\textwidth}
        \centering
        \includegraphics[height=2.0in]{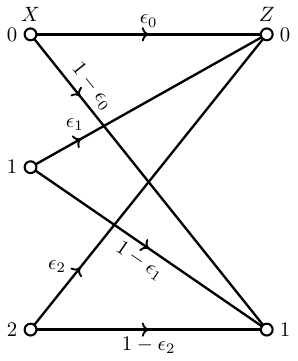}
        \caption{A ternary-input, binary-output}
        \label{fig:TIBO}
    \end{subfigure}
    ~ 
    \hspace{-5mm}\begin{subfigure}[t]{0.55\textwidth}
        \centering
        \includegraphics[height=2.2in]{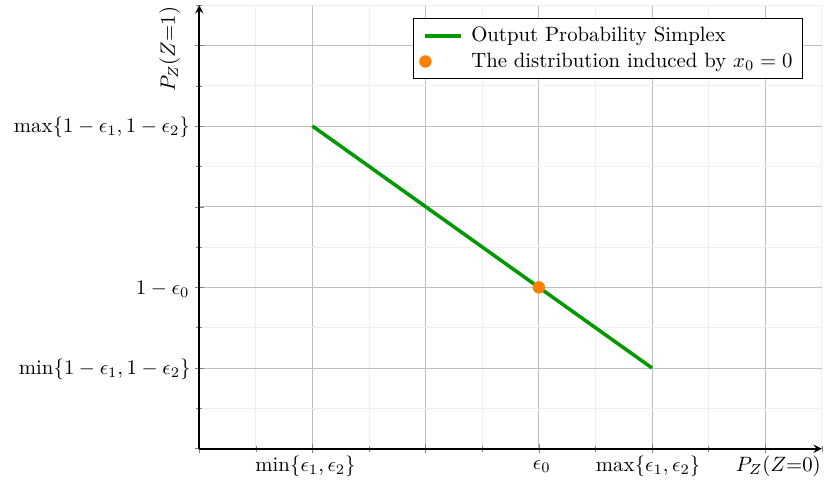}
        \caption{the set of all output distributions induced by the channel}
        \label{fig:TIBO_Dist}
    \end{subfigure}
    \caption{A general ternary-input, binary-output \ac{DMC} and the set of all output distributions induced by varying the channel input distribution}
    \label{fig:TIBO_Example}
    \vspace{-0.7cm}
\end{figure*}
\subsubsection{Ternary-Input, Binary-Output Warden \texorpdfstring{\acp{DMC}}{DMCs}}
We now enlarge the input alphabet of the binary-input, binary-output \ac{DMC} considered in the previous section to a ternary input alphabet $\calX=\br{0,1,2}$, as illustrated in Fig.~\ref{fig:TIBO}, while still 
assuming that the innocent symbol is $x_0=0$. 
In this case, the covertness constraint in \eqref{eq:Coverness_GIBO_Final} reduces to 
\begin{align}
\min\br{\epsilon_1,\epsilon_2}\le\epsilon_0\le\max\br{\epsilon_1,\epsilon_2},\label{eq:Cov_TIBO_Final}
\end{align}
Note that the constraint in \eqref{eq:Cov_TIBO_Final} is less restrictive than the constraint in \eqref{eq:Cov_BIBO_Final}. In this case, the set of all probability distributions that can be induced at the channel output in \eqref{eq:Conv_GIBO} reduces to
\begin{equation}
    \Delta_{\mathrm{TIBO}}
    \triangleq
    \mathrm{conv}\left\{
        (\epsilon_0,1-\epsilon_0),
        (\epsilon_1,1-\epsilon_1),
        (\epsilon_2,1-\epsilon_2)
    \right\}.
    \label{eq:Conv_TIBO}
\end{equation}
For fixed $\epsilon_0$ and $\epsilon_1$, the set $\Delta_{\mathrm{TIBO}}$ contains the binary-input counterpart $\Delta_{\mathrm{BIBO}}$ in \eqref{eq:Conv_BIBO}, that is, $\Delta_{\mathrm{BIBO}}\subseteq\Delta_{\mathrm{TIBO}}$. Consequently, the output distribution in the no-communication mode, namely $(\epsilon_0,1-\epsilon_0)$, may lie in the interior of $\Delta_{\mathrm{TIBO}}$, thereby enabling the achievability of a positive covert communication rate.
\begin{corollary}
For a ternary-input binary-output warden channel, 
Theorem~\ref{thm:Capacity_Classical_Sto} yields the following covert capacity region for deterministic encoding
\begin{align}
\calC_{\mathrm{C\text{-}C}} =\bigcup_{\substack{p_0,p_1\ge0:\\ p_0+p_1\le1,\\ p_1\epsilon_1+(1-p_0-p_1)\epsilon_2
    =(1-p_0)\epsilon_0}}\left.\begin{cases}(R,R_K):\\
  R\le\bbI(X;Y),\\
  R_K\ge\sbr{\bbI(X;Z)-\bbI(X;Y)}^+,
\end{cases}\hspace{-3mm}\right\},\nonumber
\end{align}where $\calX=\{0,1,2\}, p_X(0)=p_0, p_X(1)=p_1,p_X(2)=1-p_0-p_1$, and the condition $p_1\epsilon_1+(1-p_0-p_1)\epsilon_2=(1-p_0)\epsilon_0$ is the covertness constraint $p_Z=q_0$. This capacity region can be computed once the legitimate receiver's channel $W_{Y|X}$ is specified.
\end{corollary}
\subsubsection{Quaternary-Input, Binary-Output Warden \texorpdfstring{\acp{DMC}}{DMCs}}Further enlarging the channel input alphabet continues to relax the conditions under which a positive covert rate can be achieved. For example, we consider expanding the input alphabet of the ternary-input, binary-output \ac{DMC} channel depicted in  Fig.~\ref{fig:TIBO_Example} and consider a binary-output \ac{DMC} with an input alphabet of size four, as depicted in Fig.~\ref{fig:QIBO_Example} and assume that the innocent symbol is still $x_0=0$. 
For this channel, the covertness constraint in \eqref{eq:Coverness_GIBO_Final} reduces~to
\begin{align}
    \min\br{\epsilon_1,\epsilon_2,\epsilon_3}\le\epsilon_0\le\max\br{\epsilon_1,\epsilon_2,\epsilon_3}.\label{eq:Coverness_QIBO_Final}
\end{align}Note that, for fixed $\epsilon_0,\epsilon_1$, and $\epsilon_2$, the constraint in \eqref{eq:Coverness_QIBO_Final} is less restrictive than the constraints in \eqref{eq:Cov_BIBO_Final} and \eqref{eq:Cov_TIBO_Final}. In this case, the set of all probability distributions that can be induced at the channel output in \eqref{eq:Conv_GIBO} reduces to
\begin{equation}
    \Delta_{\mathrm{QIBO}}
    \triangleq
    \mathrm{conv}\left\{
        (\epsilon_0,1-\epsilon_0),
        (\epsilon_1,1-\epsilon_1),
        (\epsilon_2,1-\epsilon_2),
        (\epsilon_3,1-\epsilon_3)
    \right\}.
    \label{eq:Conv_QIBO}
\end{equation}
\begin{figure*}[t!]
    \centering
    \hspace{-5mm}\begin{subfigure}[t]{0.45\textwidth}
        \centering
        \includegraphics[height=2.0in]{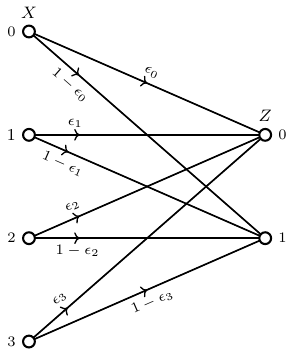}
        \caption{A quaternary-input, binary-output}
        \label{fig:QIBO}
    \end{subfigure}%
    ~ 
    \hspace{-5mm}\begin{subfigure}[t]{0.55\textwidth}
        \centering
        \includegraphics[height=2.2in]{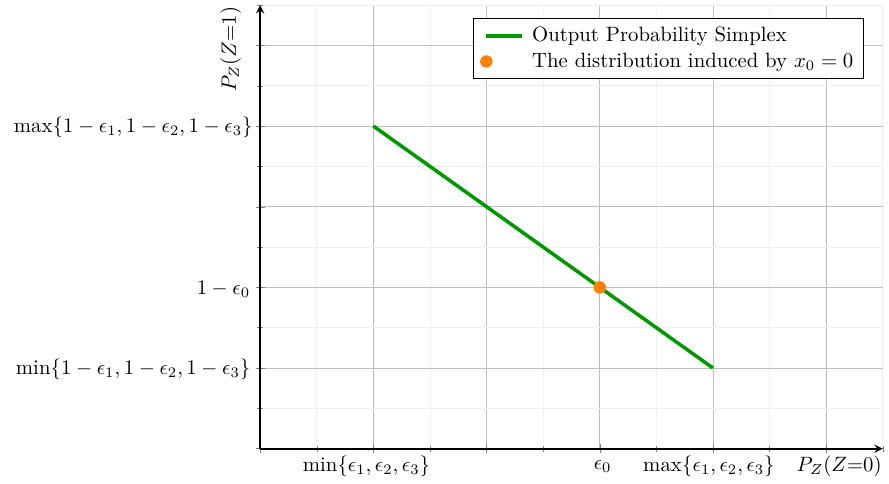}
        \caption{the set of all output distributions induced by the channel.}
        \label{fig:QIBO_Dist}
    \end{subfigure}
    \caption{A general quaternary-input, binary-output \ac{DMC} and the set of all output distributions induced by varying the channel input distribution}
    \label{fig:QIBO_Example}
    \vspace{-0.7cm}
\end{figure*}
As before, for fixed $\epsilon_0$, $\epsilon_1$, and $\epsilon_2$, the set $\Delta_{\mathrm{QIBO}}$ contains the ternary-input counterpart $\Delta_{\mathrm{TIBO}}$ in \eqref{eq:Conv_TIBO}, that is, $\Delta_{\mathrm{BIBO}}\subseteq\Delta_{\mathrm{TIBO}}\subseteq\Delta_{\mathrm{QIBO}}$. Consequently, the output distribution in the no-communication mode, namely $(\epsilon_0,1-\epsilon_0)$, may lie in the interior of $\Delta_{\mathrm{QIBO}}$, thereby enabling the achievability of a positive covert communication rate.
\subsection{Continuous-Input Binary-Output Warden Channels}
\label{sec:Continupus_Binary}
We now consider an extreme case in which the channel input alphabet is continuous and the channel output alphabet is binary, and show that the covertness constraint $p_Z=q_0$ in \eqref{eq:Capacity_Classical_S_Sto} of Theorem~\ref{thm:Capacity_Classical_Sto} can always be satisfied.

In particular, suppose that the warden observes
\begin{subequations}\label{eq:threshold_BO_Channel}
\begin{equation}
Z=\indi{1}_{\br{X+N\ge\tau}},
\label{eq:threshold_channel}
\end{equation}
where $N\sim\mathcal{N}(0,1)$ is independent of $X$ and $\tau\in\mathbb{R}$ is a fixed threshold. Since the no-communication mode corresponds to no transmission in this channel, we take the innocent symbol to be $x_0=0$ \cite{Bash13}. The input is continuous and subject to the average-power constraint
\begin{equation}
\mathbb{E}[X^2]\leq T,
\label{eq:power_constraint}
\end{equation}
where $T>0$. 
\end{subequations}
\begin{theorem}
    For the channel described in \eqref{eq:threshold_BO_Channel}, the covertness constraint $p_Z=q_0$ can always be satisfied.
\end{theorem}
\begin{proof}
For a given input $x$, the probability that the warden observes $Z=1$ is
\begin{align}
\Psi(x)
&\triangleq \bbP(Z=1\mid X=x) \nonumber\\
&=\bbP(X+N>\tau\mid X=x) \nonumber\\
&=\bbP(N>\tau-x) \nonumber\\
&=\Phi(x-\tau),
\label{eq:psi_x}
\end{align}
where $\Phi(t)\triangleq\int_{-\infty}^t\frac{1}{\sqrt{2\pi}}\exp\pr{-\frac{u^2}{2}}du$ denotes the cumulative distribution function of a standard Gaussian random variable. Hence, $W_{Z|X=x}
=
\bigl(1-\Psi(x),\Psi(x)\bigr)$ and $p_Z
=
\bigl(1-\bbE\sbr{\Psi(X)},\bbE\sbr{\Psi(X)}\bigr)$, where $\Psi(x)=\Phi(x-\tau)$. 
In particular, when the innocent symbol $x_0=0$ is transmitted, $q_0=W_{Z|X=0}
=
\bigl(1-\Psi(0),\Psi(0)\bigr)$. Therefore, since $Z$ is binary, the covertness constraint $p_Z=q_0$ is equivalent to $\bbE\sbr{\Psi(X)}=\Psi(0)$. 
Since $\Phi(\cdot)$ is strictly increasing, $\Psi(x)=\Phi(x-\tau)$ is also strictly increasing. Therefore, for any $\epsilon>0$, $\Psi(-\epsilon)<\Psi(0)<\Psi(\epsilon)$. 
It follows that there exists a unique $p_\epsilon\in(0,1)$ such that $p_\epsilon \Psi(\epsilon)
+(1-p_\epsilon)\Psi(-\epsilon)
=\Psi(0)$, where
\begin{equation}
p_\epsilon
=
\frac{\Psi(0)-\Psi(-\epsilon)}
{\Psi(\epsilon)-\Psi(-\epsilon)}.
\nonumber
\end{equation}
Thus, $\pr{1-\Psi(0),\Psi(0)}$ lies in the convex hull of
$\br{\pr{1-\Psi(-\epsilon),\Psi(-\epsilon)},\pr{1-\Psi(\epsilon),\Psi(\epsilon)}}$. Consider the input distribution
\begin{equation}
X=
\begin{cases}
\epsilon, & \text{with probability }p_\epsilon,\\
-\epsilon, & \text{with probability }1-p_\epsilon.
\end{cases}
\label{eq:two_point_distribution}
\end{equation}
The resulting output distribution at the warden satisfies
\begin{align}
p_Z(1)
&=p_\epsilon \bbP(Z=1\mid X=\epsilon)
+(1-p_\epsilon)\bbP(Z=1\mid X=-\epsilon) \nonumber\\
&=p_\epsilon \Psi(\epsilon)
+(1-p_\epsilon)\Psi(-\epsilon) \nonumber\\
&=\Psi(0).
\label{eq:covertness_exact}
\end{align}
Since $Z$ is binary, \eqref{eq:covertness_exact} implies $p_Z=W_{Z|X=0}$. 
Thus, the input distribution in \eqref{eq:two_point_distribution} induces the same distribution at the warden as the innocent symbol. 
Moreover, because both possible input symbols have magnitude $\epsilon$,
\begin{align}
\mathbb{E}[X^2]
&=p_\epsilon\epsilon^2
+(1-p_\epsilon)\epsilon^2 \nonumber\\
&=\epsilon^2.\nonumber
\end{align}
For any $T>0$, we can choose $0<\epsilon<\sqrt{T}$, which gives $\mathbb{E}[X^2]=\epsilon^2<T$. Consequently, for every positive power constraint $T>0$, there exists a non-degenerate input distribution satisfying the power constraint and inducing the same output distribution at the warden as the innocent symbol. 
\end{proof}
\begin{corollary}
    Assuming $Y=Z$, the covert capacity of the channel described in \eqref{eq:threshold_BO_Channel}, is
    \begin{align}
\calC_{\mathrm{C\text{-}C}}
&=
\bbH_b\bigl(\Phi(-\tau)\bigr)
-
\inf_{\substack{
p_X:\,\mathbb{E}[\Phi(X-\tau)]=\Phi(-\tau)\\
\mathbb{E}[X^2]\leq T
}}
\mathbb{E}\left[
\bbH_b\bigl(\Phi(X-\tau)\bigr)
\right],
\label{eq:CC_capacity_1}
\end{align}
where $\bbH_b(u)
\triangleq
-u\log_2u-(1-u)\log_2(1-u)$
denotes the binary entropy function.
\end{corollary}
\begin{proof}
Recall that $p_Z=\bigl(1-\bbE\sbr{\Psi(X)},\bbE\sbr{\Psi(X)}\bigr)$, hence, the covertness constraint becomes $\mathbb{E}[\Phi(X-\tau)]=\Phi(-\tau)$. Also, since $Y=Z$, the covertness constraint fixes the output distribution to $p_Y=W_{Z|X=0}=\bigl(1-\Phi(-\tau),\Phi(-\tau)\bigr)$. Under the average-power constraint \eqref{eq:power_constraint}, Theorem~\ref{thm:Capacity_Classical_Sto} gives
\begin{align}
\calC_{\mathrm{C\text{-}C}}
&=
\sup_{\substack{
p_X:\,\mathbb{E}[\Phi(X-\tau)]=\Phi(-\tau)\\
\mathbb{E}[X^2]\leq T
}}
\left[
\bbH(Y)-\bbH(Y|X)
\right] \nonumber\\
&=
\bbH_b\bigl(\Phi(-\tau)\bigr)
-
\inf_{\substack{
p_X:\,\mathbb{E}[\Phi(X-\tau)]=\Phi(-\tau)\\
\mathbb{E}[X^2]\leq T
}}
\mathbb{E}\left[
\bbH_b\bigl(\Phi(X-\tau)\bigr)
\right],
\label{eq:CC_capacity}
\end{align}
where the first equality follows since $p_Y(1)=\Phi(-\tau)$, and the second equality follows from $W_{Y|X=x}(1)=\Phi(x-\tau)$.
\end{proof}
For example, when $x_0=0$, $T=5$, and $\tau=1$, consider the two-point input
distribution
\begin{equation}
X=
\begin{cases}
-1.87447, & \text{with probability }0.842308,\\
+3.59714, & \text{with probability }0.157692.
\end{cases}\nonumber
\end{equation}For this distribution, $\bbE\sbr{X^2}\approx5$ and $\bbE\sbr{\Phi(X-1)}=\Phi(-1)$ therefore both the power and the covertness constraints are satisfied. For this case we have $\Phi(-1)=0.1587\Rightarrow\bbH_b\pr{0.1587}\approx0.6311$ and $\bbH_b(Y|X)=0.842308h2\Phi(-2.87447)+0.157692h2\Phi(2.59714)\approx0.024508$ and therefore $\calC_{\mathrm{C\text{-}C}}\ge0.606575$ and since $Y=Z$ no secret key is needed.
\begin{remark}[Comparison with \cite{Lee15}]
\label{rem:Lee_Work}
The example presented in this subsection shows that when the warden employs a threshold detector whose threshold is known to the transmitter, a positive covert rate can be achieved for every positive average-power constraint. This result is consistent with the results of \cite{Lee15}, where a positive communication rate is shown to be achievable when the warden employs an energy detector and has uncertainty about the noise variance of its channel. However, the underlying mechanisms are different. In \cite{Lee15}, positive-rate communication is enabled by the warden's uncertainty about its noise power, whereas in our setting, positive-rate covert communication is enabled by the ability to choose a non-degenerate input distribution that induces the same output distribution at the warden as the innocent symbol.
\end{remark}
\subsection{Extension to Warden \texorpdfstring{\acp{DMC}}{DMCs} with Arbitrary Output Alphabet}
In this section, we generalize the analysis to \acp{DMC} with an arbitrary warden output alphabet $\calZ$ of cardinality $\card{\calZ}$. 
\begin{lemma}
\label{lemma:Enlarging_Channel_Input_Alphabet}
Let $(\calX_1, W_{Z|X}^{(1)}, \calZ)$ and $(\calX_2, W_{Z|X}^{(2)}, \calZ)$ be two \acp{DMC} with the same output alphabet $\calZ$, where 
$\calX_1 \subseteq \calX_2$. Assume that for each $x \in \calX_1$ we have $W_{Z|X=x}^{(2)} = W_{Z|X=x}^{(1)}$. 
Let $\Delta_1$ and $\Delta_2$ denote the sets of all output distributions induced by varying the input distributions over $\calX_1$ and $\calX_2$, respectively. Then $\Delta_1 \subseteq \Delta_2$.
\end{lemma}
\begin{proof}
Any output distribution in $\Delta_1$ can be written as $p_Z(z) = \sum_{x \in \calX_1} p_X(x) W_{Z|X}^{(1)}(z|x)$,  
for some input distribution $p_X$ supported on $\calX_1$. Since $\calX_1 \subseteq \calX_2$, the same $p_X$ can be viewed as a valid input distribution over $\calX_2$ by assigning zero probability to symbols in $\calX_2 \setminus \calX_1$. Hence, the same output distribution belongs to $\Delta_2$, which implies $\Delta_1 \subseteq \Delta_2$. Equivalently, $\Delta_1$ is the convex hull of a subset of the points generating~$\Delta_2$. 
\end{proof}

Lemma~\ref{lemma:Enlarging_Channel_Input_Alphabet} shows that enlarging the channel input alphabet expands the set of distributions that can be induced at the warden's output. Consequently, as additional non-innocent symbols are introduced, the convex hull of the corresponding output distributions grows and may eventually include the distribution induced on the warden's channel output observations in the no-communication mode, thereby allowing the covertness constraint to be satisfied. Intuitively, when a probability simplex is mapped through a \ac{DMC} to another probability simplex of lower dimension, the images of the input extreme points (i.e., the conditional output distributions induced by input symbols) may no longer all correspond to extreme points of the resulting output simplex. Consequently, some of these output distributions may lie in the interior of the convex hull generated by the others. In particular, the output distribution corresponding to the innocent symbol may lie in the interior of this convex set, in which case the covertness constraint $p_Z = q_0$ can be satisfied.

In particular, this expansion process may reach a terminal regime in which the convex hull of the output distributions induced by non-innocent symbols coincides
with the entire probability simplex $\Delta^{\card{\calZ}-1}$.
The following lemma characterizes this regime and provides sufficient conditions under which the covertness constraint $p_Z = q_0$ can be satisfied for any choice of innocent symbol.

\begin{lemma}[Channel-Level Sufficient Condition to Satisfy Covertness Constraint]
\label{lemma:General_Sufficient_Covertness}
Consider a \ac{DMC} $(\calX, W_{Z|X}, \calZ)$ with finite output alphabet $\calZ$ and channel input alphabet $\calX$. Let $x_0\in\calX$ denote the innocent symbol, and define $v_x\triangleq W_{Z|X=x}\in\bbR^{\card{\calZ}}$. 
If the set of output distributions induced by non-innocent symbols satisfies
\begin{align}
\mathrm{conv}\bigl\{v_x : x\in\calX\setminus\{x_0\}\bigr\}
= \Delta^{|\calZ|-1},\nonumber
\end{align}
then the covertness constraint $p_Z = q_0$ can be satisfied for \emph{any} choice of innocent symbol $x_0$. 
\end{lemma}
\begin{proof}
   The proof follows immediately from the fact that $v_{x_0}\in\Delta^{|\calZ|-1}$ and $\mathrm{conv}\bigl\{v_x : x\in\calX\setminus\{x_0\}\bigr\}= \Delta^{|\calZ|-1}$.
\end{proof}

\begin{remark}[Conditions to Achieve a Positive Covert Communication Rate]
Lemma~\ref{lemma:General_Sufficient_Covertness} provides a sufficient condition for satisfying the covertness
constraint $p_Z=q_0$. However, satisfying the covertness constraint alone
does not guarantee a positive covert communication rate. A necessary
condition for a positive rate is that the legitimate receiver can
distinguish at least two channel inputs; in particular, there must exist
some $x_i\in\calX$ such that $W_{Y|X=x_i}\neq W_{Y|X=x_0}$. 
Indeed, if $W_{Y|X=x}=W_{Y|X=x_0}$, for all $x\in\calX$, then $Y$ is independent of $X$ for every input distribution, and hence $I(X;Y)=0$. Thus, although the geometry of the warden's channel determines
whether the covertness constraint can be satisfied, a positive covert
communication rate additionally requires that the legitimate receiver's
channel provide a nonzero information rate.
\end{remark}

\begin{figure}[t!]
    \centering
    \includegraphics[height=2.0in]{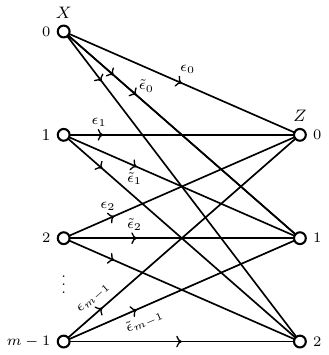}
    \caption{A general ternary-output \ac{DMC}}
    \label{fig:mO_Example}
    \vspace{-0.7cm}
\end{figure}
The preceding results show that, in general, the feasibility of the
covertness constraint is governed by the geometry of the convex hull
of the output distributions induced by the non-innocent input symbols.
In particular, enlarging the input alphabet can enlarge this convex
hull and may eventually make it possible to reproduce the output
distribution induced by the innocent symbol. Lemma~\ref{lemma:General_Sufficient_Covertness}
identifies an extreme case in which the convex hull of the
non-innocent output distributions coincides with the entire output
probability simplex, in which case covertness can be achieved for any
choice of innocent symbol.

We next illustrate these geometric phenomena for the simplest
non-binary-output alphabet, namely $\card{\calZ}=3$. In this case,
the output distributions belong to the two-dimensional simplex
$\Delta^2$, allowing us to visualize how the convex hull of the
non-innocent output distributions evolves as the channel input
alphabet is enlarged. In particular, we consider examples in which
adding a single input symbol changes the covertness constraint from
infeasible to feasible, as well as a continuous-input channel for
which the set of induced output distributions forms a curve in
$\Delta^2$.
\subsubsection{Discrete-Input Ternary-Output Warden Channels}
The geometric characterization developed above becomes particularly
transparent when the warden's output alphabet is ternary. In this
case, each input symbol $x\in\calX$ induces a probability vector $v_x=
\bigl(
W_{Z|X}(0|x),
W_{Z|X}(1|x),
W_{Z|X}(2|x)
\bigr)
\in\Delta^2$, and the set of output distributions induced by varying the input
distribution is the convex hull of these points in the two-dimensional simplex $\Delta^2$. 
This two-dimensional geometry allows us to illustrate explicitly the
effect of enlarging the channel input alphabet on the feasibility of
the covertness constraint.

Consider a \ac{DMC} with input alphabet $\calX$, where $\abs{\calX}=m$, and output alphabet $\calZ=\br{0,1,2}$, where $W_{Z|X}(0|i)=\epsilon_i$, $W_{Z|X}(1|i)=\tilde{\epsilon}_i$, and $W_{Z|X}(2|i)=1-\epsilon_i-\tilde{\epsilon}_i$, for $i\in[0\!:\!m-1]$ and $\epsilon_i,\tilde{\epsilon}_i\in[0,\!1]$. See Fig.~\ref{fig:mO_Example}. Assume that the innocent symbol is $x_0=j$ with $j\in[0\!:\!m-1]$. Define
% \begin{subequations}\label{eq:QITO_Vs}
\begin{align}
v_i &\triangleq W_{Z|X=i}
      = \bigl(\epsilon_i,\tilde{\epsilon}_i,1-\epsilon_i-\tilde{\epsilon}_i\bigr), \label{eq:QITO_Vi}
\end{align}
Each $v_i$ is a probability vector in the simplex $\Delta^2 \triangleq \left\{ (p_0,p_1,p_2)\in\mathbb{R}_+^3 : p_0+p_1+p_2=1 \right\}$. For this channel, the covertness constraint $p_Z=q_0$ can be satisfied by a non-degenerate input distribution if and only~if
\begin{align}
    v_j \in \operatorname{conv}\big\{v_i : i\in[0\!:\!m-1]\setminus\!\br{j}\!\big\},
    \label{eq:Conv_mITO}
\end{align}that is, if $v_j$ lies in the convex hull of the output distributions induced by non-innocent symbols. 
\begin{remark}[Designing the Innocent Symbol]
\label{rem:innocent_symbol_gen}
If the innocent symbol can be chosen freely and $\abs{\calZ}\geq 3$,
then, unlike the binary-output case stated in
Theorem~\ref{thm:General_DMC}, having
$\abs{\calX}>\abs{\calZ}$ alone does not guarantee that there exists an
innocent symbol for which the covertness constraint in
\eqref{eq:Conv_mITO} can be satisfied. Nevertheless, appropriately
choosing the innocent symbol may facilitate satisfying the covertness
constraint for some channels.

In particular, if there exists an input symbol $x_j\in\calX$ such that \eqref{eq:Conv_mITO} holds, then choosing $x_0=x_j$ allows the output distribution induced by the innocent symbol to be reproduced by a non-trivial input distribution.
Thus, choosing such an input symbol as the innocent symbol can satisfy the covertness constraint.
\end{remark}
When the innocent symbol $x_0$ is fixed, the condition in \eqref{eq:Conv_mITO} can become less restrictive as the cardinality $m$ of the channel input alphabet increases, i.e., as additional input symbols are introduced. To illustrate this, we consider quaternary-input, ternary-output and quinary-input, ternary-output \acp{DMC} in the following.
\subsubsection{Quaternary- and Quinary-Input, Ternary-Output Warden \texorpdfstring{\acp{DMC}}{DMCs}}
When the channel input alphabet of the \ac{DMC} illustrated in Fig.~\ref{fig:mO_Example} is $\calX=\br{0,1,2,3}$ and the innocent symbol is $x_0=0$. 
The covertness constraint $p_Z = q_0$ in \eqref{eq:Conv_mITO} reduces to
\begin{align}
    v_0&\in\mathrm{conv}\br{v_1,v_2,v_3}.
    \label{eq:Conv_QITO}
\end{align}

Next, we enlarge the input alphabet to $\calX=\br{0,1,2,3,4}$ while keeping the output alphabet $\calZ=\br{0,1,2}$ and the innocent symbol fixed at $x_0=0$. 
For this expanded channel, the covertness constraint becomes
\begin{align}
    v_0&\in\mathrm{conv}\br{v_1,v_2,v_3,v_4}.
    \label{eq:Conv_5ITO}
\end{align}
For a fixed $\epsilon_0,\tilde{\epsilon}_0,\epsilon_1,\tilde{\epsilon}_1,\epsilon_2,\tilde{\epsilon}_2,\epsilon_3$, and $\tilde{\epsilon}_3$, since $\mathrm{conv}\br{v_1,v_2,v_3}\subseteq\mathrm{conv}\br{v_1,v_2,v_3,v_4}$, the covertness constraint \eqref{eq:Conv_5ITO} is more relaxed compared with the constraint in \eqref{eq:Conv_QITO}, and enlarging the input alphabet can only relax the covertness constraint. Moreover, if $v_4\notin\mathrm{conv}\br{v_1,v_2,v_3}$, the inclusion is strict, and the feasible set strictly enlarges. To illustrate this phenomenon, consider the numerical choice $\epsilon_0=0.4,\tilde{\epsilon}_0=0.15,\epsilon_1=0.1,\tilde{\epsilon}_1=0.8,\epsilon_2=0.6,\tilde{\epsilon}_2=0.2,\epsilon_3=0.2,\tilde{\epsilon}_3=0.6,\epsilon_4=\tfrac{13}{30}\approx0.4333$, and $\tilde{\epsilon}_4=\tfrac{7}{90}\approx0.0778$. For this example, it can be verified that, $v_0\notin\mathrm{conv}\br{v_1,v_2,v_3}$ while $v_0\in\mathrm{conv}\br{v_1,v_2,v_3,v_4}$. This geometric relationship is depicted in Fig.~\ref{fig:Prob_Simplex_45I}.

This example illustrates that enlarging the channel input alphabet can strictly enlarge the set of output distributions that can be induced at the warden. In particular, the addition of a single non-innocent input
symbol can cause the innocent output distribution to enter the convex hull
of the output distributions induced by non-innocent symbols, thereby
changing the covertness constraint from infeasible to feasible. This
observation motivates the continuous-input setting considered next, where
the set of available input symbols is no longer finite and the resulting
convex-hull geometry can exhibit qualitatively different behavior.
\begin{figure}
\centering
\vspace{-0.2cm}
\includegraphics[height=2.5in]{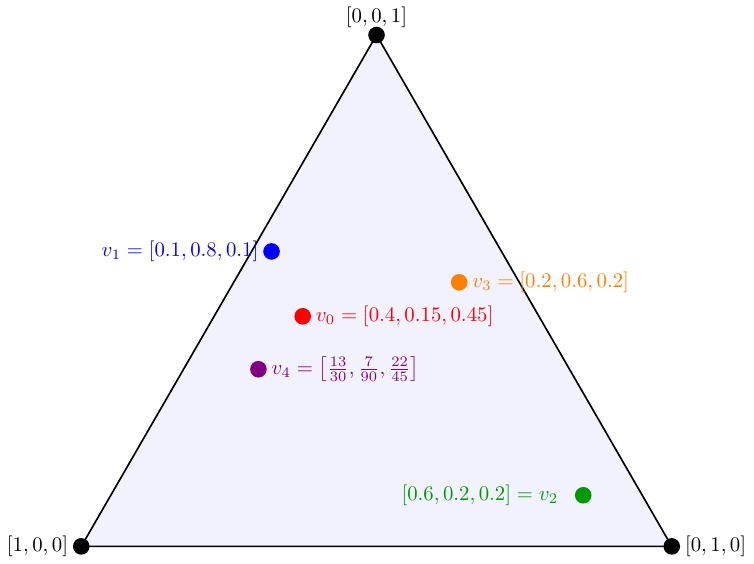}
\caption{The vectors $v_0,v_1,v_2,v_3,v_4$ represented in the ternary probability simplex. In particular, $v_0=0.1v_1+0.9v_4$, and hence $v_0\in\operatorname{conv}\{v_1,v_2,v_3,v_4\}$, but $v_0\notin\operatorname{conv}\{v_1,v_2,v_3\}$}
\label{fig:Prob_Simplex_45I}
\end{figure}
\subsubsection{Continuous-Input Ternary-Output Warden Channels}
We consider a channel with a continuous input alphabet and a discrete warden output alphabet of size $\lvert\calZ\rvert=3$. The warden observes a quantized version of $X+N$ according to
\begin{align}
Z=
\begin{cases}
0, & X+N<\tau_1,\\
1, & \tau_1\le X+N<\tau_2,\\
2, & X+N\ge\tau_2,
\end{cases}
\label{eq:Channel_Cont_Input_Disc_Output}
\end{align}
where $\tau_1<\tau_2$, and $N$ follows a uniform distribution over $[-1,\!1]$ and is independent of $X$. The \ac{PDF} and the \ac{CDF} of the random variable $N$ are, respectively, given by
\begin{align}
f_N(n)&=
\begin{cases}
\frac{1}{2}, & -1<n<1,\\
0, & \text{otherwise},
\end{cases}\nonumber\\
    F_N(n)&=\begin{cases}
0, & n\le-1,\\
\frac{n+1}{2}, & -1< n<1,\\
1, & n\ge1.
\end{cases}\nonumber
\end{align}
Therefore, the warden's channel transition probabilities are given by
\begin{align}
W_{Z|X}(0|x) &= \bbP(x+N<\tau_1)=\begin{cases}
0, & x\ge\tau_1+1,\\
\frac{\tau_1-x+1}{2}, & \tau_1-1< x<\tau_1+1,\\
1, & x\le\tau_1-1,
\end{cases},\nonumber\\
W_{Z|X}(2|x) &= \bbP(x+N\ge\tau_2)
              =\begin{cases}
0, & x\le\tau_2-1,\\
\frac{x-\tau_2+1}{2}, & \tau_2-1< x<\tau_2+1,\\
1, & x\ge\tau_2+1,
\end{cases},\nonumber\\
W_{Z|X}(1|x) &= \bbP(\tau_1\le x+N<\tau_2)=1-W_{Z|X}(0|x)-W_{Z|X}(2|x).\nonumber
\end{align}Equivalently, $W_{Z|X}(1|x)$ is the length of the intersection $[\tau_1-x,\!\tau_2-x]\cap[-1,\!1]$ divided by 2. 
Therefore, as seen in Fig.~\ref{fig:Continuous_Input_Ternary}, if we choose the channel input such that $x\le\tau_1-1$, we can induce the corner point $v_x=[1,0,0]$, while choosing $x\ge\tau_2+1$ induces $v_x=[0,0,1]$.  Moreover, if there exists an input satisfying $\tau_1+1\le x\le\tau_2-1$, then we can induce the corner point $v_x=[0,1,0]$. Such an input exists if and only if $\tau_2-\tau_1\ge2$. Therefore, if
$\tau_2-\tau_1\ge2$, all three corner points
$[1,0,0]$, $[0,1,0]$, and $[0,0,1]$ can be induced, and hence $\operatorname{conv}\{v_x:x\in\mathbb{R}\}=\Delta^2$. Furthermore, if $\tau_2-\tau_1>2$, each of the three corner points can be induced by infinitely many input symbols. Hence, for any choice of innocent symbol $x_0$, removing $x_0$ does not eliminate any of the three corner points, and therefore $\operatorname{conv}\{v_x:x\ne x_0\}=\Delta^2$. Thus, by Lemma~\ref{lemma:General_Sufficient_Covertness}, the covertness constraint $p_Z=q_0$ can be satisfied using a non-innocent input distribution.
\begin{figure}
\centering
\vspace{-0.2cm}
\includegraphics[height=0.8in]{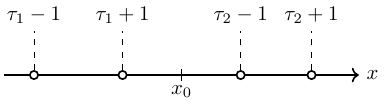}
\caption{Locations of the boundary points induced by the two quantization thresholds and the innocent symbol $x_0$}
\label{fig:Continuous_Input_Ternary}
\end{figure}

It is important to note, however, that the condition $\tau_2-\tau_1\ge2$ is sufficient but not necessary for satisfying the covertness constraint. For example, let
$x_0=0$, $\tau_1=-0.5$, and $\tau_2=0.5$. In this case, $\tau_2-\tau_1=1<2$, and hence the three corner points of $\Delta^2$ cannot all be induced. In this case,
\begin{align}
    v_0&=\sbr{W_{Z|X}(0|0),W_{Z|X}(1|0),W_{Z|X}(2|0)}\nonumber\\
    &=\sbr{F_N(\tau_1),F_N(\tau_2)-F_N(\tau_1),1-F_N(\tau_2)}.\label{eq:v0_ternary_Unif}
\end{align}
Hence, from \eqref{eq:v0_ternary_Unif} we have $v_0=\left[\frac{1}{4},\frac{1}{2},\frac{1}{4}\right]$. For the two non-innocent input symbols $x_1=-0.5$ and $x_2=0.5$, we have $v_{-0.5}=\left[\frac{1}{2},\frac{1}{2},0\right]$ and $v_{0.5}=\left[0,\frac{1}{2},\frac{1}{2}\right]$. Therefore, $v_0
=\frac{1}{2}v_{-0.5}+\frac{1}{2}v_{0.5}$. 
Thus, the covertness constraint can be satisfied even though the three corner points cannot all be induced.

\section{Extension to \texorpdfstring{\acp{MAC}}{MACs}}
\label{sec:MAC}

The geometric perspective developed for point-to-point channels can be extended naturally to \acp{MAC}. In particular, a \ac{MAC} provides a natural mechanism for enlarging the effective channel input alphabet. Although each transmitter has an input alphabet $\calX_i$, $i\in\{1,2\}$, the joint channel input belongs to $\calX_1\times\calX_2$. Thus, a \ac{MAC} can provide additional joint input symbols whose induced output distributions may be used to reproduce the warden's no-communication output distribution.
\begin{figure}[t!]
    \centering
        \includegraphics[height=2.0in]{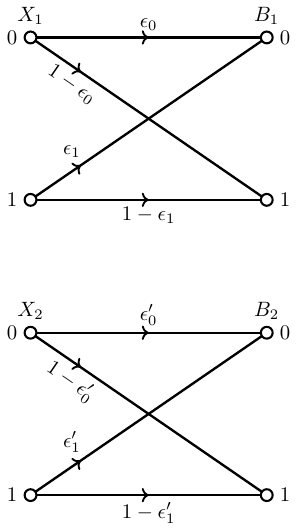}
    \caption{A binary-input, binary-output discrete memoryless \ac{MAC} with the channel output $Z=B_1\oplus B_2$}
    \label{fig:bibo_MAC_Example}
    \vspace{-0.7cm}
\end{figure}
\begin{subequations}\label{eq:BIBO_MAC}
To illustrate this idea, consider the following binary-input \ac{MAC}, as illustrated in Fig.~\ref{fig:bibo_MAC_Example}. Let
\begin{equation}
    B_1=X_1\oplus N_1,
    \qquad
    B_2=X_2\oplus N_2,
\end{equation}
where $X_1,X_2\in\{0,1\}$, and $\oplus$ denotes addition modulo two, and suppose that the warden observes
\begin{equation}
    Z=B_1\oplus B_2.
\end{equation}
\end{subequations}
For the first component, define $\Pr(B_1=0|X_1=i)=\epsilon_i$, for $i\in\{0,1\}$ and, similarly, for the second component, $\Pr(B_2=0|X_2=j)=\epsilon'_j$, for $j\in\{0,1\}$. Consequently, the warden's channel satisfies
\begin{align}
    a_{ij}
    &\triangleq W_{Z|X_1X_2}(0|i,j)\nonumber\\
    &=
    \epsilon_i\epsilon'_j
    +(1-\epsilon_i)(1-\epsilon'_j).\label{eq:aij_defi}
\end{align}
Let $\ell\in\{0,1\}$ and $\ell'\in\{0,1\}$ denote the innocent symbols of the two transmitters, respectively. The no-communication distribution at the warden is therefore $q_0(z)
    =
    W_{Z|X_1X_2}(z|\ell,\ell')$, and, since $Z$ is binary, $q_0(0)=a_{\ell,\ell'}$. 

\subsection{\texorpdfstring{\acp{MAC}}{MACs} with Independent Channel Inputs}
\label{subsec:MAC_independent}
We first consider a \ac{MAC} in which each transmitter covertly communicates its own private message, and hence the channel inputs are independent. The following capacity result from \cite[Theorem~3]{ExtendedPaperMAC} specializes to this setting.
\begin{theorem}[\!\!{\cite[Theorem~3]{ExtendedPaperMAC}}]
    \label{thm:Capacity_Classical_MAC}
    % Let
\begin{subequations}\label{eq:Capacity_Classical_All_MAC}
The covert capacity of the classical \ac{DMC} \ac{MAC} $W_{YZ\lvert X_1X_2}$, under the deterministic encoding is
\begin{align}
\calC_{\mathrm{C\text{-}MAC}} =\bigcup_{p_{X_1X_2YZ}\in\calB}\left.\begin{cases}(R_1,R_2):\\
  R_1\le\bbI(X_1;Y|X_2),\\
  R_2\le\bbI(X_2;Y|X_1),\\
  R_1+R_2\le\bbI(X_1,X_2;Y),
\end{cases}\hspace{-3mm}\right\},
\label{eq:Capacity_Classical_MAC}
\end{align}
where
\begin{align}
  \calB \triangleq \left.\begin{cases}p_{X_1X_2YZ}:\\
p_{X_1X_2YZ}=p_{X_1}p_{X_2}W_{YZ\lvert X_1X_2},\\
\bbI(X_1;Y|X_2)\ge\bbI(X_1;Z),\\
\bbI(X_2;Y|X_1)\ge\bbI(X_2;Z),\\
\bbI(X_1,X_2;Y)\ge\bbI(X_1,X_2;Z),\\
p_Z=q_0\triangleq W_{Z|X_1=x_0^{(1)},X_2=x_0^{(2)}},\\
\end{cases}\right\}.\label{eq:Capacity_Classical_S_MAC}
\end{align}
\end{subequations}
\end{theorem}
Similar to the point-to-point setting, we focus on conditions under which the covertness constraint $p_Z=q_0$ in Theorem~\ref{thm:Capacity_Classical_MAC} can be satisfied. The mutual information constraints in \eqref{eq:Capacity_Classical_S_MAC} can be circumvented by using a secret key shared among the legitimate terminals. Define $\alpha_1
    \triangleq
    \Pr(X_1\neq\ell)$, and $
    \alpha_2
    \triangleq
    \Pr(X_2\neq\ell')$. 
The resulting joint input distribution is
\begin{align*}
    \Pr(X_1=\ell,X_2=\ell')
    &=(1-\alpha_1)(1-\alpha_2),\\
    \Pr(X_1=1-\ell,X_2=\ell')
    &=\alpha_1(1-\alpha_2),\\
    \Pr(X_1=\ell,X_2=1-\ell')
    &=(1-\alpha_1)\alpha_2,\\
    \Pr(X_1=1-\ell,X_2=1-\ell')
    &=\alpha_1\alpha_2.
\end{align*}
The set of warden output distributions that can be induced by independent inputs is
\begin{subequations}
\begin{equation}
    \Delta_{\mathrm{ind-MAC}}
    =
    \left\{
        \left(
        q_Z^{(\alpha_1,\alpha_2)}(0),1-q_Z^{(\alpha_1,\alpha_2)}(0)
        \right):
        (\alpha_1,\alpha_2)\in[0,1]^2
    \right\},
\end{equation}where
\begin{align}
    q_Z^{(\alpha_1,\alpha_2)}(0)
    &=
    (1-\alpha_1)(1-\alpha_2)a_{\ell,\ell'}
    +\alpha_1(1-\alpha_2)a_{1-\ell,\ell'}
    \nonumber\\
    &\quad
    +(1-\alpha_1)\alpha_2a_{\ell,1-\ell'}
    +\alpha_1\alpha_2a_{1-\ell,1-\ell'}.
\label{eq:MAC_induced_output}
\end{align}
\end{subequations}
Thus, the covertness constraint under independent channel inputs is equivalent to asking whether the no-communication distribution $q_0$ belongs to $\Delta_{\mathrm{ind-MAC}}$. More specifically, we seek
\begin{equation}
    q_Z^{(\alpha_1,\alpha_2)}=q_0\label{eq:Covertness_Ind_MAC}
\end{equation}
for a non-trivial pair $(\alpha_1,\alpha_2)\neq(0,0)$,  with at least one of $\alpha_1,\alpha_2$ strictly between zero and one.

Define
\begin{align*}
    d_{10}
    &\triangleq
    a_{1-\ell,\ell'}-a_{\ell,\ell'},\\
    d_{01}
    &\triangleq
    a_{\ell,1-\ell'}-a_{\ell,\ell'},\\
    d_{11}
    &\triangleq
    a_{1-\ell,1-\ell'}-a_{\ell,\ell'}.
\end{align*}
Then \eqref{eq:MAC_induced_output} gives
\begin{equation}
    q_Z^{(\alpha_1,\alpha_2)}(0)-q_0(0)
    =
    F(\alpha_1,\alpha_2),
\end{equation}
where
\begin{equation}
    F(\alpha_1,\alpha_2)
    =
    \alpha_1(1-\alpha_2)d_{10}
    +(1-\alpha_1)\alpha_2d_{01}
    +\alpha_1\alpha_2d_{11}.
\label{eq:MAC_covertness_function}
\end{equation}
Therefore, the covertness constraint in \eqref{eq:Covertness_Ind_MAC} is equivalent to
\begin{equation}
    F(\alpha_1,\alpha_2)=0.
\label{eq:MAC_covertness_condition}
\end{equation}
We say that an input distribution is \emph{non-trivial} if at least one transmitter has a non-degenerate input distribution, so that at least one transmitter can potentially convey information. In particular, for the independent-input model considered above, this means $\alpha_1\in(0,1)$ or $\alpha_2\in(0,1)$. This definition does not require both transmitters to be active. The following theorem characterizes the existence of a non-trivial independent input distribution satisfying the covertness constraint.

\begin{theorem}[Independent-input \ac{MAC}]
\label{prop:MAC_independent_non-trivial}
There exists a non-trivial independent input distribution satisfying the covertness constraint if and only if at least one of the following conditions holds:
\begin{subequations}\label{eq:MAC_indi_condition}
\begin{align}
    d_{10}&=0,
\label{eq:MAC_d10_zero}\\
    d_{01}&=0,
\label{eq:MAC_d01_zero}\\
    \min\{d_{10},d_{01},d_{11}\}
    &<0<
    \max\{d_{10},d_{01},d_{11}\}.
\label{eq:MAC_strict_sign_condition}
\end{align}
\end{subequations}
Equivalently, there exists a non-trivial independent input distribution satisfying the covertness constraint if and only if
\begin{subequations}\label{eq:MAC_independent_non-trivial_equivalent}
\begin{align}
    a_{1-\ell,\ell'}&=a_{\ell,\ell'},\label{eq:MAC_independent_non-trivial_equivalent_1}\\
    a_{\ell,1-\ell'}&=a_{\ell,\ell'},\label{eq:MAC_independent_non-trivial_equivalent_2}\\
    \min_{\substack{i,j\in\{0,1\}\\
    (i,j)\neq(\ell,\ell')}}
    a_{ij}
    &<
    a_{\ell,\ell'}
    <
    \max_{\substack{i,j\in\{0,1\}\\
    (i,j)\neq(\ell,\ell')}}
    a_{ij}.\label{eq:MAC_independent_non-trivial_equivalent_3}
\end{align}
\end{subequations}
Furthermore, when the innocent symbols $x_0^{(1)}$ and $x_0^{(2)}$ can be freely chosen, the covertness constraint $p_Z=q_0$ can always be satisfied.
\end{theorem}

\begin{proof}
We first prove necessity. Suppose that there exists a non-trivial $(\alpha_1,\alpha_2)\in[0,1]^2$ satisfying $F(\alpha_1,\alpha_2)=0$, where $\alpha_1\in(0,1)$ or $\alpha_2\in(0,1)$. If $d_{10}>0$, $d_{01}>0$, and $ d_{11}>0$ then all three terms in $F(\alpha_1,\alpha_2)$ are nonnegative. Moreover, for any non-trivial $(\alpha_1,\alpha_2)\neq(0,0)$, at least one of the three coefficients $\alpha_1(1-\alpha_2)$, $(1-\alpha_1)\alpha_2$, and $\alpha_1\alpha_2$ is strictly positive. Hence $F(\alpha_1,\alpha_2)>0$, which is a contradiction. The same argument applies if $d_{10}<0$, $d_{01}<0$, and $d_{11}<0$. Now suppose that there is no strict sign change among $d_{10},d_{01},d_{11}$. If either $d_{10}=0$ or $d_{01}=0$, then the first or second condition in Theorem~\ref{prop:MAC_independent_non-trivial} holds. Otherwise, $d_{10}$ and $d_{01}$ are both nonzero and have the same strict sign. Since there is no strict sign change, the only possibility for $F(\alpha_1,\alpha_2)$ to vanish is $d_{11}=0$. In this case,
\begin{equation*}
    F(\alpha_1,\alpha_2)
=
\alpha_1(1-\alpha_2)d_{10}
+(1-\alpha_1)\alpha_2d_{01}.
\end{equation*}For any non-trivial $(\alpha_1,\alpha_2)\neq(1,1)$, at least one of $\alpha_1(1-\alpha_2)$ and $(1-\alpha_1)\alpha_2$ is strictly positive. Since $d_{10}$ and $d_{01}$ have the same strict sign, $F(\alpha_1,\alpha_2)$ has that same strict sign and therefore cannot equal zero. The remaining point $(1,1)$ is not non-trivial, since both input distributions are deterministic. Thus, a non-trivial
solution can exist only if
\begin{equation*}
d_{10}=0,
\quad\text{or}\quad
d_{01}=0,
\quad\text{or}\quad
\min\{d_{10},d_{01},d_{11}\}<0<
\max\{d_{10},d_{01},d_{11}\}.
\end{equation*}

We next prove sufficiency. First, suppose that $d_{10}=0$. Choose any $\alpha_1\in(0,1)$ and $\alpha_2=0$.  Then $F(\alpha_1,0)=\alpha_1d_{10}=0$, and the resulting input distribution is non-trivial. Similarly, if $d_{01}=0$, choose any $\alpha_2\in(0,1)$, and $\alpha_1=0$, for which $F(0,\alpha_2)=\alpha_2d_{01}=0$. 
It remains to consider
\begin{align*}
    \min\{d_{10},d_{01},d_{11}\}<0<
\max\{d_{10},d_{01},d_{11}\}.
\end{align*}
If $d_{10}$ and $d_{11}$ have opposite signs, set $\alpha_1=1$ and $\alpha_2=\frac{d_{10}}{d_{10}-d_{11}}$. Then $0<\alpha_2<1$ and
\begin{align*}
F(1,\alpha_2)
=
(1-\alpha_2)d_{10}
+\alpha_2d_{11}
=0.
\end{align*}
Hence, the resulting input distribution is non-trivial.  Likewise, if $d_{01}$ and $d_{11}$ have opposite signs, set $\alpha_2=1$ and $\alpha_1=\frac{d_{01}}{d_{01}-d_{11}}$. Since $d_{01}$ and $d_{11}$ have opposite signs, $0<\alpha_1<1$, and
\begin{align*}
F(\alpha_1,1)
=
(1-\alpha_1)d_{01}
+\alpha_1d_{11}
=0.
\end{align*}
Finally, suppose that $d_{10}$ and $d_{01}$ have opposite signs while $d_{11}=0$. Choose $\alpha_2=\frac{1}{2}$, and $\alpha_1=\frac{d_{01}}{d_{01}-d_{10}}$. Since $d_{10}$ and $d_{01}$ have opposite signs, $0<\alpha_1<1$. Furthermore,
\begin{align*}
F\left(\alpha_1,\frac{1}{2}\right)
&=
\frac{1}{2}\alpha_1d_{10}
+\frac{1}{2}(1-\alpha_1)d_{01} \notag\\
&=
\frac{1}{2}
\left[
\alpha_1d_{10}
+(1-\alpha_1)d_{01}
\right]
=0.
\end{align*}
Thus, in all cases satisfying
\begin{align*}
d_{10}=0,
\quad\text{or}\quad
d_{01}=0,
\quad\text{or}\quad
\min\{d_{10},d_{01},d_{11}\}<0<
\max\{d_{10},d_{01},d_{11}\},
\end{align*}
there exists a non-trivial product input distribution satisfying the covertness constraint. This completes the proof.

Now, suppose that the innocent symbols $x_0^{(1)}$ and $x_0^{(2)}$ can be freely chosen. Consider the four values $a_{00},a_{10},a_{01},a_{11}$, and let $a_{\ell,\ell'}$ be the second-smallest value among them. If it is strictly between the smallest and largest of the remaining three values, \eqref{eq:MAC_independent_non-trivial_equivalent_3} holds. Otherwise, the second-smallest value is attained by at least two distinct input pairs. Choosing one of these pairs as the innocent pair, another distinct pair induces the same warden output distribution, and either \eqref{eq:MAC_independent_non-trivial_equivalent_1} or \eqref{eq:MAC_independent_non-trivial_equivalent_2} holds. Hence, there always exists a choice of innocent symbols for which the covertness constraint can be satisfied with a non-trivial joint input distribution.
\end{proof}
\begin{remark}[Geometric interpretation]
Theorem~\ref{prop:MAC_independent_non-trivial} has a useful geometric interpretation. The four joint input symbols induce the four warden output distributions $W_{Z|X_1X_2}(\cdot|i,j)$, for $ (i,j)\in\{0,1\}^2$. If arbitrary joint input distributions were allowed, the resulting warden output distributions would fill the convex hull of these four points. Under independent inputs, however, the probabilities assigned to the four joint input symbols must satisfy the product constraint.

For the binary-output \ac{MAC} considered here, the resulting set of warden output distributions is nevertheless the same interval as the convex hull, since $q_Z^{(\alpha_1,\alpha_2)}(0)$ is a continuous bilinear function of $(\alpha_1,\alpha_2)$ and its extrema over $[0,1]^2$ occur at the four corner points. The distinction between independent and joint inputs therefore lies not in the set of reachable warden output distributions, but in the input distributions that realize a given warden output distribution. In particular, the covertness condition depends not only on whether the innocent output distribution lies in the convex hull of the non-innocent output distributions, but also on whether that point is reachable under the product-distribution restriction in a way that allows at least one transmitter to use a non-degenerate input distribution.
\end{remark}

\begin{remark}[Comparison with the point-to-point channel]
The conditions under which the covertness constraint can be satisfied
with a non-trivial input distribution in the binary-input, binary-output
\ac{MAC} can be less restrictive than those for the corresponding
binary-input, binary-output point-to-point channels studied in
Section~\ref{sec:DIS_Inp_BO}. In particular, condition
\eqref{eq:MAC_independent_non-trivial_equivalent} can be satisfied even
when the corresponding point-to-point condition \eqref{eq:Cov_BIBO_Final}
cannot be satisfied. This difference arises from the additional joint
input symbols available in the \ac{MAC}, which can produce warden output
distributions on opposite sides of the innocent output distribution and
thereby enable cancellation.

For example, let $\epsilon_0=0.9$, $\epsilon_1=0.6$, $\epsilon'_0=0.55$, and $\epsilon'_1=0.95$. Then, from \eqref{eq:aij_defi}, $a_{00}=0.54$, $a_{10}=0.51$, $a_{01}=0.86$, $a_{11}=0.59$, and, for $\ell=\ell'=0$, $d_{10}=-0.03$, $d_{01}=0.32$, and $d_{11}=0.05$. Thus, $\min\{d_{10},d_{01},d_{11}\}
=-0.03<0<
0.32
=\max\{d_{10},d_{01},d_{11}\}$, and hence the condition in
\eqref{eq:MAC_strict_sign_condition} is satisfied. In fact, choosing $\alpha_1=0.5$, $\alpha_2=0.075$
gives $F(0.5,0.075)=0$, and therefore satisfies the covertness constraint in \eqref{eq:MAC_covertness_condition} with a non-trivial product input
distribution.

In contrast, when either transmitter is absent, the resulting channel reduces to a point-to-point channel for which the condition in Section~\ref{sec:DIS_Inp_BO} is not satisfied, and hence a positive covert rate cannot be achieved under the corresponding point-to-point model. This example illustrates that the interaction between the two transmitters can create a covert communication opportunity that is unavailable when either transmitter operates alone.
\end{remark}

\subsubsection{Examples} In this section, we consider two cases. In the first, the channel associated with one of the transmitters is a \ac{BSC}; in the second, the channels from both transmitters to the warden are \acp{BSC}.

\paragraph{One \texorpdfstring{\ac{BSC}}{BSC} Component}

Now fix $\ell=\ell'=0$ and suppose that the second component is a BSC,
so that $\epsilon'_1=1-\epsilon'_0$, while the first component is arbitrary. In this case,
\begin{align*}
    d_{10}
    &=
    (\epsilon_1-\epsilon_0)(2\epsilon'_0-1),\\
    d_{01}
    &=
    -(2\epsilon_0-1)(2\epsilon'_0-1),\\
    d_{11}
    &=
    -(\epsilon_1-\epsilon_0)(2\epsilon'_0-1)
    =-d_{10}.
\end{align*}
Therefore, whenever $d_{10}\neq0$, $d_{10}$ and $d_{11}$ have opposite
signs. Theorem~\ref{prop:MAC_independent_non-trivial} then guarantees the existence of a non-trivial independent input distribution satisfying the covertness constraint.

For example, setting $\alpha_1=1$ and $\alpha_2=\frac{d_{10}}{d_{10}-d_{11}}=\frac{1}{2}$ gives $F(1,1/2)=0$. Thus, Transmitter~2 can use a non-degenerate input distribution while Transmitter~1 is fixed to its non-innocent symbol. The degenerate cases can be handled directly. In particular, if $d_{10}=0$ and $d_{01}\neq0$, choosing $\alpha_1\in(0,1)$ and $\alpha_2=0$ satisfies the covertness constraint.

\paragraph{Both Components are \texorpdfstring{\acp{BSC}}{BSCs}}

Finally, suppose that both component channels are \acp{BSC}: $\epsilon_1=1-\epsilon_0$ and $\epsilon'_1=1-\epsilon'_0$. Then $a_{00}=a_{11}=\epsilon_0\epsilon'_0+(1-\epsilon_0)(1-\epsilon'_0)$ and $a_{10}=a_{01}=\epsilon_0(1-\epsilon'_0)+(1-\epsilon_0)\epsilon'_0$. For $\ell=\ell'=0$, $d_{11}=0$ and $d_{10}=d_{01}=-(2\epsilon_0-1)(2\epsilon'_0-1)$. For non-degenerate BSCs, $d_{10}=d_{01}\neq0$. Hence
\begin{equation*}
    F(\alpha_1,\alpha_2)
    =
    d_{10}
    \left[
    \alpha_1(1-\alpha_2)
    +(1-\alpha_1)\alpha_2
    \right].
\end{equation*}
The two terms in brackets are nonnegative and can simultaneously vanish only when $(\alpha_1,\alpha_2)=(0,0)$ or $(\alpha_1,\alpha_2)=(1,1)$. The first corresponds to the all-innocent input, while the second is a deterministic non-innocent input. Neither provides a non-degenerate input distribution for communication. Consequently, when both component channels are non-degenerate BSCs, no transmitter can communicate using a non-trivial independent input distribution while satisfying the  covertness constraint.
\subsection{\texorpdfstring{\acp{MAC}}{MACs} with Degraded Message Sets}
\label{subsec:MAC_degraded_message_sets}
We next consider covert communication over a \ac{MAC} with degraded message sets, in which the second transmitter may have access to the message transmitted by the first transmitter \cite{ExtendedPaperMAC}. In this setting, the channel inputs may be correlated.
\begin{theorem}[\!\!{\cite[Theorem~3]{ExtendedPaperMAC}}]
    \label{thm:Capacity_Classical_MAC_DMS}
\begin{subequations}\label{eq:Capacity_Classical_All_MAC_DMS}
The covert capacity of the classical discrete memoryless \ac{MAC} $W_{YZ\lvert X_1X_2}$ with degraded message sets, under the deterministic encoding is
\begin{align}
\calC_{\mathrm{C\text{-}MAC}} =\bigcup_{p_{X_1X_2YZ}\in\calB}\left.\begin{cases}(R_1,R_2):\\
  R_2\le\min\{\bbI(X_2;Y|X_1),\bbI(X_1,X_2;Y)-\bbI(X_1;Z)\},\\
  R_1+R_2\le\bbI(X_1,X_2;Y),\\
\end{cases}\hspace{-3mm}\right\},
\label{eq:Capacity_Classical_MAC_DMS}
\end{align}
where
\begin{align}
  \calB \triangleq \left.\begin{cases}p_{X_1X_2YZ}:\\
p_{X_1X_2YZ}=p_{X_1X_2}W_{YZ\lvert X_1X_2},\\
\bbI(X_1,X_2;Y)\ge\bbI(X_1,X_2;Z),\\
p_Z=q_0\triangleq W_{Z|X_1=x_0^{(1)},X_2=x_0^{(2)}},\\
\end{cases}\right\}.\label{eq:Capacity_Classical_S_MAC_DMS}
\end{align}
\end{subequations}
\end{theorem}
Similar to the previous section, we study the conditions under which the covertness constraint $p_Z=q_0$ in Theorem~\ref{thm:Capacity_Classical_MAC_DMS} can be satisfied. The essential distinction from the preceding setting is that the joint input distribution need not factor as $p_{X_1}p_{X_2}$. Let $\Pr(X_1=1,X_2=0)=\alpha$, $\Pr(X_1=0,X_2=1)=\beta$, $\Pr(X_1=1,X_2=1)=\eta$, and $\Pr(X_1=0,X_2=0)=1-\alpha-\beta-\eta$, where $0\le\alpha,\beta,\eta\le1$. Therefore, the induced warden output distribution is
\begin{align*}
    q_Z(0)&=\alpha W_{Z|X_1X_2}(0|1,0)+\beta W_{Z|X_1X_2}(0|0,1)+\eta W_{Z|X_1X_2}(0|1,1)+(1-\alpha-\beta-\eta) W_{Z|X_1X_2}(0|0,0)\nonumber\\
    &=\alpha a_{10}+\beta a_{01}+\eta a_{11}+(1-\alpha-\beta-\eta) a_{00}.
\end{align*}
Therefore, the set of all warden output distributions that can be induced by arbitrary joint input distributions is
\begin{equation}
    \Delta_{\mathrm{DMS-MAC}}
    =
    \operatorname{conv}
    \left\{
        W_{Z|X_1X_2}(\cdot|i,j):
        (i,j)\in\{0,1\}^2
    \right\}.
\label{eq:MAC_joint_convex_hull}
\end{equation}
Thus, unlike the independent-input case, every point in the convex hull of the four conditional warden output distributions can be realized by a suitable joint input distribution.

The covertness constraint requires $q_Z=q_0$ or, equivalently,
\begin{equation}
    \alpha a_{1-\ell,\ell'}+\beta a_{\ell,1-\ell'}+\eta a_{1-\ell,1-\ell'}+(1-\alpha-\beta-\eta) a_{\ell,\ell'}=a_{\ell,\ell'}.
\label{eq:MAC_joint_covertness}
\end{equation}

We seek a non-trivial joint input distribution, meaning one for which at least one transmitter has a non-degenerate marginal input distribution and can therefore potentially convey information.
\begin{theorem}[Binary-input, binary-output \ac{MAC} with degraded message set]
\label{prop:MAC_joint_non-trivial}
There exists a non-trivial joint input distribution satisfying the covertness constraint if and only if
\begin{align}
\min\{d_{10},d_{01},d_{11}\}
    &\leq 0 \leq
    \max\{d_{10},d_{01},d_{11}\}.
\label{eq:MAC_range_condition_DMS}
\end{align}
Equivalently, there exists a non-trivial joint input distribution satisfying the  covertness constraint if and only if
\begin{equation}
    \min_{\substack{i,j\in\{0,1\}\\
    (i,j)\neq(\ell,\ell')}}
    a_{ij}
    \leq
    a_{\ell,\ell'}
    \leq
    \max_{\substack{i,j\in\{0,1\}\\
    (i,j)\neq(\ell,\ell')}}
    a_{ij}.
\label{eq:MAC_joint_non-trivial_condition}
\end{equation}
Furthermore, when the innocent symbols $x_0^{(1)}$ and $x_0^{(2)}$ can be freely chosen, the covertness constraint $p_Z=q_0$ can always be satisfied with a non-trivial joint input distribution.
\end{theorem}

\begin{proof}
Suppose first that $a_{\ell,\ell'}=a_{i^\star j^\star}$ for some $(i^\star,j^\star)\neq(\ell,\ell')$. For any $\theta\in(0,1)$, consider the joint input distribution $\Pr(X_1=\ell,X_2=\ell')=1-\theta$ and $\Pr(X_1=i^\star,X_2=j^\star)=\theta$. Then $p_Z(0)=(1-\theta)a_{\ell,\ell'}+\theta a_{i^\star j^\star}=a_{\ell,\ell'}$, and therefore the covertness constraint is satisfied. Since the two joint input symbols are distinct, at least one of the two transmitters has a non-degenerate marginal input distribution. Hence, the input distribution is non-trivial. Now suppose that
\begin{equation*}
    \min_{\substack{i,j\\(i,j)\neq(\ell,\ell')}}
    a_{ij}
    <
    a_{\ell,\ell'}
    <
    \max_{\substack{i,j\\(i,j)\neq(\ell,\ell')}}
    a_{ij}.
\end{equation*}
Then there exist two non-innocent joint input pairs
$(i_1,j_1)$ and $(i_2,j_2)$ such that $a_{i_1j_1}<a_{\ell,\ell'}<a_{i_2j_2}$. Consequently, there exists $\theta\in(0,1)$ satisfying $(1-\theta)a_{i_1j_1} +\theta a_{i_2j_2}=a_{\ell,\ell'}$. Choosing $\Pr(X_1=i_1,X_2=j_1)=1-\theta$ and $\Pr(X_1=i_2,X_2=j_2)=\theta$ therefore yields the covertness constraint. Since the two joint input pairs are distinct, at least one transmitter has a non-degenerate marginal input distribution, and hence the input distribution is non-trivial.

Conversely, suppose that $a_{\ell,\ell'}
<
\min_{\substack{i,j\\(i,j)\neq(\ell,\ell')}}a_{ij}$. Then every non-innocent joint input pair produces a strictly larger probability of $Z=0$ than the innocent pair. Hence, any joint input distribution assigning positive probability to a non-innocent pair satisfies $p_Z(0)>a_{\ell,\ell'}$, and the covertness constraint cannot be satisfied. The case $a_{\ell,\ell'}>\max_{\substack{i,j\\(i,j)\neq(\ell,\ell')}}a_{ij}$ is analogous. Therefore, the stated range condition is necessary.

Finally, suppose that the innocent symbols $x_0^{(1)}$ and $x_0^{(2)}$ can be freely chosen. Consider the four values $a_{00}$, $a_{10}$, $a_{01}$, and $a_{11}$, and let $a_{\ell,\ell'}$ be the second-smallest value among these four values. If the selected value is strictly between the smallest and largest of the remaining three values, then \eqref{eq:MAC_joint_non-trivial_condition} is satisfied. Otherwise, the second-smallest value is attained by at least two distinct joint input pairs. Choosing one of these pairs as the innocent pair, another distinct pair induces the same warden output distribution, and the first case of the proof applies. Hence, there always exists a choice of innocent symbols for which the covertness constraint can be satisfied with a non-trivial joint input distribution.
\end{proof}

\begin{remark}[Effect of increasing the channel input alphabets]
    Similar to the previous section, increasing the channel input alphabet of either or both transmitters introduces additional channel input pairs and, consequently, additional warden output distributions that can be used to satisfy the covertness constraint. Thus, enlarging the input alphabets can further relax the conditions in Theorems~\ref{prop:MAC_independent_non-trivial} and \ref{prop:MAC_joint_non-trivial}.
\end{remark}

\paragraph{Comparison of Independent and Joint Inputs}

Theorems~\ref{prop:MAC_independent_non-trivial} and \ref{prop:MAC_joint_non-trivial} reveal a fundamental distinction between independent and statistically dependent channel inputs. For the joint-input model, the covertness constraint with a non-trivial input distribution is possible whenever the innocent warden output probability $a_{\ell,\ell'}$ belongs to the closed interval generated by the three non-innocent output probabilities:
\begin{equation}
    \min_{\substack{i,j\\(i,j)\neq(\ell,\ell')}}
    a_{ij}
    \leq
    a_{\ell,\ell'}
    \leq
    \max_{\substack{i,j\\(i,j)\neq(\ell,\ell')}}
    a_{ij}.
\end{equation}
This is precisely the usual convex-hull condition.

For independent inputs, the same closed range condition is not sufficient. If the innocent output probability coincides only with $a_{1-\ell,\ell'}$, or only with $a_{\ell,1-\ell'}$, a non-trivial independent input distribution can still be constructed by allowing the corresponding transmitter to vary while the other remains innocent. However, if the innocent output probability coincides only with $a_{1-\ell,1-\ell'}$, independence prevents the transmitters from mixing only the two joint input symbols $(\ell,\ell')$ and $(1-\ell,1-\ell')$. In that case, the remaining two joint input symbols necessarily appear whenever both marginal inputs are randomized, and their contributions cannot be canceled unless the condition in Theorem~\ref{prop:MAC_independent_non-trivial} is satisfied.

Thus, the joint-input model has a less restrictive condition for the existence of a non-trivial covert input distribution because it permits the transmitters to exploit correlation between their channel inputs. In particular, consider again the case in which both component channels are non-degenerate BSCs and $\ell=\ell'=0$. Then $a_{00}=a_{11}$, and $a_{10}=a_{01}\neq a_{00}$. Theorem~\ref{prop:MAC_independent_non-trivial} shows that the covertness condition cannot be satisfied by a non-trivial independent input distribution: neither $d_{10}$ nor $d_{01}$ is zero, while $d_{11}=0$ and the three differences do not contain both a strictly positive and a strictly negative value. Hence no transmitter can communicate using a non-trivial independent input distribution while satisfying the covertness constraint. In contrast, for the joint-input model, consider $\Pr(X_1=0,X_2=0)=1-\theta$ and $\Pr(X_1=1,X_2=1)=\theta$, for any $\theta\in(0,1)$. Since $a_{00}=a_{11}$,
\begin{equation*}
    p_Z(0)=(1-\theta)a_{00}+\theta a_{11}=a_{00}=q_0(0).
\end{equation*}Thus, the joint input distribution is covert and non-trivial. This example demonstrates that statistical dependence between the transmitters can create a covert communication opportunity that is unavailable under independent channel inputs.

\section{Classical-Quantum and Quantum Channels}
\label{sec:Quantum}
In this section, we study classical covert communication over point-to-point and \ac{MAC} classical-quantum and quantum channels. We begin by characterizing an achievable rate region for covert capacity for transmitting classical information over point-to-point classical-quantum channels and then extend the results developed for classical channels in Section~\ref{sec:Classical} to the classical-quantum setting. We next extend the results developed for classical \acp{MAC} in Section~\ref{sec:MAC} to the classical-quantum settings. Subsequently, we derive an achievable rate region for covert communication over quantum channels. Finally, we identify conditions under which positive covert rates can be achieved over quantum channels.
\subsection{Classical-Quantum Channels}
\label{sec:Classical_Quantum_Channels}
In this section, we consider classical-quantum channels and extend the results developed for classical channels in Section~\ref{sec:Classical} to the classical-quantum setting.
\begin{theorem}[Achievable Rate Region for  Classical-Quantum Channels]
\label{thm:capacity_cq}
An inner bound on the covert capacity of the classical-quantum channel $\pr{\calX,\br{\rho_{YZ}^{(x)}}_{x\in\calX}}$ with stochastic encoding~is
\begin{subequations}\label{eq:Cap_cq}
\begin{align}%
&\calC_{\mathrm{C\text{-}CQ}}\supseteq\bigcup_{\rho_{UXYZ}\in\calG} \left.\begin{cases}(R,R_K):\\
  R<\bbI(U;Y),\\
  R_K>\sbr{\bbI(U;Z)-\bbI(U;Y)}^+,
\end{cases}\hspace{-3mm}\right\},
\label{eq:inRnone_cq}
\end{align}where
\begin{align}
  \calG \triangleq \left.\begin{cases}\rho_{UXYZ}:\\
  \rho_{UXYZ}=p_{UX}(u,x)\rho_{YZ}^{(x)},\\
\sigma_0=\tra_Y\sbr{\rho_{YZ}^{(x_0)}},\\
\rho_Z=\sigma_0,\\
\end{cases}\hspace{-3mm}\right\}.\label{eq:thm_S_cq}
\end{align} 
\end{subequations}
\end{theorem}
The proof of Theorem~\ref{thm:capacity_cq} follows from that of Theorem~\ref{thm:Achievable_Quantum} in Appendix~\ref{proof:thm:Achievable_Quantum}.

Similar to the classical case discussed in Remark~\ref{remark:Sto_Deter}, Theorem~\ref{thm:capacity_cq} characterizes an achievable rate region on the covert capacity under stochastic encoding, which is implemented via channel prefixing. One can show that this achievable rate region is no smaller than, and can be strictly larger than, the covert capacity under deterministic encoding, which is recovered by setting $U=X$ in Theorem~\ref{thm:capacity_cq}. 
In the sequel, for simplicity, we restrict our attention to the deterministic encoding scheme. Nevertheless, the results can be extended to the stochastic encoding scheme presented in Theorem~\ref{thm:capacity_cq}.

It is well known that the covert capacity of classical-quantum point-to-point channels obeys the square-root law \cite{WangCQ,Bullock25}. This behavior arises because, for many channels, the covertness constraint $\rho_Z = \sigma_0$ in \eqref{eq:thm_S_cq} of Theorem~\ref{thm:capacity_cq} cannot be satisfied. Consider a classical-quantum channel $x\in\calX\mapsto\rho_{YZ}^{(x)}$, where $\calX$ is the channel input alphabet, that maps a classical channel input $x\in\calX$ to output states $\rho_Y^{(x)}=\tra_Z\sbr{\rho_{YZ}^{(x)}}\in \calD(\calH_Y)$ at the legitimate receiver and $\rho_Z^{(x)}=\tra_Y\sbr{\rho_{YZ}^{(x)}}\in \calD(\calH_Z)$ at the warden. The covertness constraint $\rho_Z = \sigma_0$ can be satisfied if and only if the output state observed by the warden under the innocent input symbol $x_0\in\calX$, denoted by $\sigma_0$, can be expressed as a convex combination of the output states corresponding to the channel input states $x \in \calX\backslash\br{x_0}$, i.e.,
\begin{align}
    \sigma_0\in\mathrm{conv}\br{\rho_Z^{(x)}:x \in \calX\backslash\br{x_0}}.\label{eq:Covertness_Positivity_CQ}
\end{align} 
In the following, we generalize the results obtained in Section~\ref{sec:Classical} for classical channels to classical-quantum channels. Intuitively, when a probability simplex is passed through a classical \ac{DMC} and mapped into another probability simplex of smaller dimension, the images of the input extreme points (that is, the conditional output distributions induced by the input symbols) may no longer all correspond to extreme points of the resulting output simplex. As a result, some of these output distributions can end up in the interior of the convex hull formed by the others. In particular, the output distribution associated with the innocent symbol may lie in the interior of this convex set, in which case the covertness constraint $p_Z = q_0$ can be met. 
For classical-quantum \acp{DMC}, the mapping from input distributions (a probability simplex) to output states produces a convex set of density operators. If the resulting set of output states has affine dimension strictly smaller than that of the input probability simplex, i.e., $\card{\calX}-1$, then the states $\br{\rho^{(x)}}_{x\in\calX}$ may not all be affinely independent. Consequently, some of these states may lie in the convex hull of the others. In particular, the output state corresponding to the innocent symbol may lie in the convex hull of the states induced by the non-innocent symbols, in which case the covertness constraint $\rho_Z = \sigma_0$ can be satisfied.

For example, consider the following classical-quantum channel, which is the classical-quantum counterpart of the classical channel depicted in Fig.~\ref{fig:GIBO},
\begin{align}
    \rho^{(i)}_Z&=(1-\epsilon_i)\Psi_0+\epsilon_i\Psi_1,\quad\text{for}\quad i\in\sbr{0\!:\!m-1},\label{eq:BIBO_CQ}
\end{align}where $\Psi_0$ and $\Psi_1$ are two fixed quantum states that may not commute in general. Throughout this subsection, we assume that $\Psi_0\ne\Psi_1$, since otherwise the covertness constraint is automatically satisfied. Note that, for this channel, $\card{\calX}-1=m-1$, while $\mathrm{affdim}\!\br{\rho^{(i)}_Z\!:i\in[0\!:\!m-1]}=1$. Indeed, assuming that the $\epsilon_i$'s are not all equal, all of the output states lie on the line segment connecting $\Psi_0$ and $\Psi_1$. In this case, assuming that $x_0=j$ for some $j\in[0\!:\!m-1]$, the covertness constraint $\rho_Z=\sigma_0$ can be satisfied if $\sigma_0\in\mathrm{conv}\br{\rho^{(i)}_Z,i\in[0\!:\!m-1],i\ne j}$, where $\sigma_0=\rho^{(j)}_Z$. 
\begin{theorem}
    For the classical-quantum channel \eqref{eq:BIBO_CQ}, assuming that $x_0=j$, for some $j\in\sbr{0\!:\!m-1}$, there exists a non-degenerate input distribution satisfying the covertness constraint $\rho_Z=\sigma_0$, if and only~if
\begin{align}
    \min\limits_{\substack{i\in[0:m-1]\\i\ne j}}\epsilon_i\le\epsilon_j\le\max\limits_{\substack{i\in[0:m-1]\\i\ne j}}\epsilon_i.\nonumber
\end{align}
Furthermore, when $m\geq 3$ and the innocent symbol $x_0$ can be freely chosen, the covertness constraint $\rho_Z=\sigma_0$ can always be satisfied.
\end{theorem}
\begin{proof}
Since the warden's output state induced by an input distribution
$p_X$ is given by
\begin{align}
\rho_Z
&=\sum_{i=0}^{m-1}p_i\rho_Z^{(i)}\nonumber\\
&=\sum_{i=0}^{m-1}p_i
\left[(1-\epsilon_i)\Psi_0+\epsilon_i\Psi_1\right]\nonumber\\
&=
\left(1-\sum_{i=0}^{m-1}p_i\epsilon_i\right)\Psi_0
+
\left(\sum_{i=0}^{m-1}p_i\epsilon_i\right)\Psi_1,
\label{eq:CQ_GIBO_Output}
\end{align}
the output state $\rho_Z$ depends on $p_X$ only through the
quantity $\sum_{i=0}^{m-1}p_i\epsilon_i$. In particular,
\begin{equation*}
\sigma_0
=
\rho_Z^{(j)}
=
(1-\epsilon_j)\Psi_0+\epsilon_j\Psi_1.
\end{equation*}
Since $\Psi_0\neq\Psi_1$, the representation of a state as an
affine combination of $\Psi_0$ and $\Psi_1$ is unique. Hence,
$\rho_Z=\sigma_0$ if and only if
\begin{equation}
\sum_{i=0}^{m-1}p_i\epsilon_i=\epsilon_j.
\label{eq:CQ_GIBO_Covertness_Scalar}
\end{equation}
Note that the constraint in
\eqref{eq:CQ_GIBO_Covertness_Scalar} is identical to the corresponding
constraint for the classical channel in \eqref{eq:Coverness_GIBO}. 
Therefore, the conditions under which the covertness constraint
$\sigma_0=\rho_Z$ is satisfied coincide with those for its classical
counterpart in \eqref{eq:Cov_BIBO_Final}. Importantly, this equivalence
holds regardless of whether $\Psi_0$ and $\Psi_1$ commute.
\end{proof}
\subsection{Classical-Quantum \texorpdfstring{\acp{MAC}}{MACs} with Independent Inputs}
We now show that the analysis in Section~\ref{sec:MAC} for \acp{MAC} can also extend to classical-quantum \acp{MAC}. The following capacity result from \cite[Theorem~3]{ExtendedPaperMAC} specializes to a classical-quantum \ac{MAC} in which each transmitter covertly communicates its own private message, and
hence the channel inputs are independent.
\begin{theorem}[Covert Capacity of Classical-Quantum \acp{MAC}]
\label{thm:capacity_cq_MAC}
The covert capacity of the classical-quantum \ac{MAC} $\pr{\calX_1,\calX_2,\br{\rho_{YZ}^{(x_1,x_2)}}_{(x_1,x_2)\in\calX_1\times\calX_2}}$ with deterministic encoding~is
\begin{align*}
&\calC_{\mathrm{C\text{-}CQ}}=\bigcup_{\rho_{X_1X_2YZ}\in\calG} \left.\begin{cases}(R_1,R_2)\!:\\
   R_1\le\bbI(X_1;Y|X_2),\\
  R_2\le\bbI(X_2;Y|X_1),\\
  R_1+R_2\le\bbI(X_1,X_2;Y),
\end{cases}\hspace{-3mm}\right\},
\end{align*}where
\begin{align*}
  \calG \triangleq \left.\begin{cases}\rho_{X_1X_2YZ}\!:\\
\rho_{X_1X_2YZ}=p_{X_1}(x_1)p_{X_2}(x_2)\rho_{YZ}^{(x_1,x_2)},\\
\bbI(X_1;Y|X_2)\ge\bbI(X_1;Z),\\
\bbI(X_2;Y|X_1)\ge\bbI(X_2;Z),\\
\bbI(X_1,X_2;Y)\ge\bbI(X_1,X_2;Z),\\
\sigma_0=\tra_Y\sbr{\rho_{YZ}^{\pr{x_0^{(1)},\,x_0^{(2)}}}},\\
\rho_Z=\sigma_0,\\
\end{cases}\hspace{-3mm}\right\}.
\end{align*}
\end{theorem}
Similar to the previous sections, we characterize conditions under which the covertness constraint $\rho_Z=\sigma_0$ in Theorem~\ref{thm:capacity_cq_MAC} can be satisfied. Consider a classical-quantum \ac{MAC} with binary inputs $X_1,X_2\in\br{0,1}$, and suppose that the channel states observed by the warden are
\begin{equation}
\rho_Z^{(i,j)}
=(1-a_{ij})\Psi_0+a_{ij}\Psi_1,
\qquad i,j\in\br{0,1},
\label{eq:CQ_MAC_warden_states}
\end{equation}
where $\Psi_0,\Psi_1\in\calD(\calH_Z)$ are two distinct quantum states and $a_{ij}\in[0,1]$. This structure is the classical-quantum counterpart of the classical \ac{MAC} considered in Section~\ref{subsec:MAC_independent}. Let $(\ell,\ell')$ denote the innocent input pair. Let $\alpha_1=P(X_1\neq\ell)$ and $\alpha_2=P(X_2\neq\ell')$. The warden's output state is then given by
\begin{align}
\rho_Z
={}&(1-\alpha_1)(1-\alpha_2)\rho_Z^{(\ell,\ell')}+\alpha_1(1-\alpha_2)\rho_Z^{(1-\ell,\ell')}+(1-\alpha_1)\alpha_2\rho_Z^{(\ell,1-\ell')}+\alpha_1\alpha_2\rho_Z^{(1-\ell,1-\ell')}.
\label{eq:CQ_MAC_independent_output}
\end{align}
The innocent warden state is $\sigma_0=\rho_Z^{(\ell,\ell')}$. Using \eqref{eq:CQ_MAC_warden_states}, the state in \eqref{eq:CQ_MAC_independent_output} can be written as $\rho_Z=(1-\bar a)\Psi_0+\bar a\Psi_1$, where
\begin{align*}
\bar{a}
={}&(1-\alpha_1)(1-\alpha_2)a_{\ell,\ell'}+\alpha_1(1-\alpha_2)a_{1-\ell,\ell'}+(1-\alpha_1)\alpha_2a_{\ell,1-\ell'}+\alpha_1\alpha_2a_{1-\ell,1-\ell'}.
\end{align*}
Similarly, $\sigma_0=(1-a_{\ell,\ell'})\Psi_0+a_{\ell,\ell'}\Psi_1$. Since $\Psi_0\neq\Psi_1$, we have
\begin{equation*}
(1-a)\Psi_0+a\Psi_1=(1-b)\Psi_0+b\Psi_1
\quad\Longleftrightarrow\quad
a=b.
\end{equation*}
Consequently,
\begin{equation}
\rho_Z=\sigma_0
\quad\Longleftrightarrow\quad
\bar a=a_{\ell,\ell'}.\label{eq:Covertness_MAC_Ind_CQ}
\end{equation}
Define
% \begin{subequations}
\begin{align*}
d_{10}&=a_{1-\ell,\ell'}-a_{\ell,\ell'},\\
d_{01}&=a_{\ell,1-\ell'}-a_{\ell,\ell'},\\
d_{11}&=a_{1-\ell,1-\ell'}-a_{\ell,\ell'}.
\end{align*}
% \end{subequations}
Then from \eqref{eq:Covertness_MAC_Ind_CQ} the covertness condition $\rho_Z=\sigma_0$ is equivalent to
\begin{equation}
\alpha_1(1-\alpha_2)d_{10}
+(1-\alpha_1)\alpha_2d_{01}
+\alpha_1\alpha_2d_{11}=0.\label{eq:CQ_MAC_covertness_condition}
\end{equation}

Equation~\eqref{eq:CQ_MAC_covertness_condition} is identical to the covertness constraint in \eqref{eq:MAC_covertness_condition} obtained for the corresponding classical \ac{MAC} with independent inputs. Therefore, the existence of a non-trivial independent input distribution satisfying the covertness constraint is characterized by exactly the same conditions as in the classical setting. In particular, the conditions in Theorem~\ref{prop:MAC_independent_non-trivial} directly carry over to the classical-quantum \ac{MAC} under the structure in \eqref{eq:CQ_MAC_warden_states}. 

Similarly, the analysis of the other classical channels considered in Sections~\ref{sec:Classical} and~\ref{sec:MAC} can be extended to their classical-quantum counterparts. In each case, the condition under which the covertness constraint $\rho_Z=\sigma_0$ can be satisfied coincides with the corresponding classical condition $p_Z=q_0$. This observation implies that enlarging the channel input alphabet, equivalently, increasing the dimension of the input probability simplex, expands the set of output states that can be induced at the warden's observation. Consequently, as additional non-innocent symbols are introduced, the convex hull of the corresponding output states grows and may eventually contain the state induced at the warden's output under the no-communication mode, thereby enabling the covertness constraint to be satisfied. 
The following lemma generalizes Lemma~\ref{lemma:Enlarging_Channel_Input_Alphabet} from classical channels to classical-quantum channels, formalizing this intuition by showing that enlarging the channel input alphabet increases the set of output states that can be induced at the channel output.
\begin{lemma}
\label{lemma:Enlarging_Channel_Input_Alphabet_CQ}
Let $(\calX_1, \{\rho^{(x)}_1\}_{x\in\calX_1}, \calH)$ and $(\calX_2, \{\rho^{(x)}_2\}_{x\in\calX_2}, \calH)$ be two discrete memoryless classical-quantum channels with the same channel output Hilbert space $\calH$, where 
$\calX_1 \subseteq \calX_2$. Assume that
\begin{equation*}
    \rho^{(x)}_2 = \rho^{(x)}_1, \quad \forall x \in \calX_1.
\end{equation*}
Let $\calS_1$ and $\calS_2$ denote the sets of all output states induced by varying the input distributions over $\calX_1$ and $\calX_2$, respectively. Then $\calS_1 \subseteq \calS_2$, and therefore $\mathrm{affdim}(\calS_1)\le\mathrm{affdim}(\calS_2)$.
\end{lemma}
\begin{proof}
Any output state in $\calS_1$ can be written as
\begin{equation*}
    \rho_Z = \sum_{x \in \calX_1} p_X(x) \rho^{(x)}_1,
\end{equation*}for some input distribution $p_X$ supported on $\calX_1$. Since $\calX_1 \subseteq \calX_2$, the same $p_X$ can be viewed as a valid input distribution over $\calX_2$ by assigning zero probability to symbols in $\calX_2 \setminus \calX_1$. Hence, the same output state belongs to $\calS_2$, which implies $\calS_1 \subseteq \calS_2$. 
\end{proof}
The expansion of the set of output states induced at the channel output, achieved by enlarging the channel input alphabet, may reach a terminal regime in which the convex hull of the output states corresponding to non-innocent symbols coincides with the entire $\calD(\calH)$, i.e., the set of all quantum states on the Hilbert space $\calH$. The following lemma, which generalizes Lemma~\ref{lemma:General_Sufficient_Covertness} from classical channels to classical-quantum channels, characterizes this regime and provides sufficient conditions under which the covertness constraint $\rho_Z = \sigma_0$ can be satisfied for any choice of innocent state.
\begin{lemma}[Channel-Level Sufficient Condition to Satisfy Covertness Constraint]
\label{lemma:General_Sufficient_Covertness_CQ}
Consider a discrete memoryless classical-quantum channel $(\calX, \{\rho^{(x)}\}_{x\in\calX}, \calH)$ with finite output Hilbert space $\calH$ and channel input alphabet $\calX$. Let $x_0\in\calX$ be any innocent symbol. If 
\begin{align}
\mathrm{conv}\bigl\{\rho^{(x)} : x\in\calX\setminus\{x_0\}\bigr\}
= \calD(\calH),\nonumber
\end{align}
then the covertness constraint $\rho_Z = \sigma_0$ can be satisfied for \emph{any} choice of innocent symbol $x_0$. 
\end{lemma}
\begin{proof}
   The proof follows immediately from the fact that $\sigma_0\in\calD(\calH)$ and $\mathrm{conv}\bigl\{\rho^{(x)} : x\in\calX\setminus\{x_0\}\bigr\}= \calD(\calH)$.
\end{proof}
\subsection{Quantum Channels}
\label{sec:Quantum_Channels}
In this section, we study quantum-input, quantum-output channels. As in the previous sections, we begin by deriving a general achievable rate region.
\begin{theorem}[Achievable Rate Region for Quantum Channels]
\label{thm:Achievable_Quantum}
An inner bound on the covert capacity of a quantum point-to-point channel $\calN_{X\to YZ}$ is
\begin{align*}
  \calC_{\mathrm{C-Q}} \supseteq\bigcup_{\rho_{UXYZ}\in\calG}\left.\begin{cases}(R,R_K):\\
  R<\bbI(U;Y),\\
  R_K>\sbr{\bbI(U;Z)-\bbI(U;Y)}^+,
\end{cases}\hspace{-3mm}\right\},
\end{align*}where
\begin{align}
  \calG \triangleq \left.\begin{cases}\rho_{UXYZ}:\\
  \rho_{UXYZ}=\sum\limits_{u}p_U(u)\den{u}{u}_U\otimes\theta_X^{(u)}\otimes\calN_{X\to YZ}\pr{\theta_X^{(u)}},\\
\rho_Z=\sigma_0,
\end{cases}\hspace{-3mm}\right\}.\label{eq:S_Quantum}
\end{align}
\end{theorem}
The proof of Theorem~\ref{thm:Achievable_Quantum} is provided in Appendix~\ref{proof:thm:Achievable_Quantum}.

It is well known that the covert capacity of quantum point-to-point channels obeys the square-root law \cite{Bash_15}. This behavior arises because, for many channels, the covertness constraint $\rho_Z = \sigma_0$ in \eqref{eq:S_Quantum} of Theorem~\ref{thm:Achievable_Quantum} cannot be satisfied; a notable example is the bosonic channel \cite{Bash_15,Wang23}. Consider a quantum channel $\calN_{X\to YZ}$ that maps an input state $\rho_X \in \calD(\calH_X)$ to output states $\rho_Y \in \calD(\calH_Y)$ at the legitimate receiver and $\rho_Z \in \calD(\calH_Z)$ at the warden. The covertness constraint $\rho_Z = \sigma_0$ can be satisfied, and hence a positive covert communication rate is achievable, if the output state observed by the warden under the innocent input  state $\rho_0$, denoted by $\sigma_0$, can be expressed as a convex combination of the output states corresponding to the channel input states $\rho_X \in \calD(\calH_X)$ and $\rho_X\ne\rho_0$, i.e.,
\begin{align}
    \sigma_0\in\mathrm{conv}\br{\tra_Y\sbr{\calN_{X\to YZ}(\rho_X)}:\rho_X\in\calD(\calH_X)\backslash\br{\rho_0}}.\label{eq:Covertness_Positivity}
\end{align} 
\begin{lemma}[Designing the Innocent Input State]
\label{lemma:Designing_Innocent_Quantum}
If the innocent input state $\rho_0$ is a mixed state, then the covertness constraint \eqref{eq:Covertness_Positivity} can always be satisfied.
\end{lemma} 
\begin{proof}
    If the innocent input state $\rho_0$ is a mixed state, then it can be expressed as a convex combination $\rho_0=\sum_i p_i \rho_i$, where $p_i>0$ and $\sum_i p_i=1$, and $\rho_i \neq \rho_0$. Therefore, we have
    \begin{align}
        \sigma_0&=\tra_Y\sbr{\calN_{X\to YZ}(\rho_0)}\nonumber\\
        &=\tra_Y\sbr{\calN_{X\to YZ}\pr{\sum_i p_i \rho_i}}\nonumber\\
        &\mathop=\limits^{(a)}\sum_i p_i\tra_Y\sbr{\calN_{X\to YZ}\pr{ \rho_i}},\label{eq:Mixed_State_Covertness}
    \end{align}where $(a)$ follows from the linearity of both the trace operator and the quantum channel. Note that the \ac{RHS} of~\eqref{eq:Mixed_State_Covertness} implies that the channel output in the no-communication mode can be expressed as a convex combination of the outputs corresponding to input states $\rho_i$. Consequently, the covertness constraint~\eqref{eq:Covertness_Positivity} is satisfied.
\end{proof}
\begin{example}
Consider a quantum channel from a qubit input system $X$ to a pair of qubit
systems $Y$ and $Z$, defined by the isometry
\begin{align}
    V_{\gamma}\triangleq\begin{cases}
\ket{0}_X\mapsto
\ket{0}_Y\ket{0}_Z,\\
\ket{1}_X\mapsto
\sqrt{1-\gamma}\ket{1}_Y\ket{0}_Z
+
\sqrt{\gamma}\ket{0}_Y\ket{1}_Z,
\end{cases}
\end{align}
where $\gamma\in(0,1)$. The corresponding quantum channel is $\mathcal{N}_{X\to YZ}(\rho_X)
=
V_{\gamma}\rho_XV_{\gamma}^{\dagger}$. The marginal channels to the legitimate receiver and the warden are
\begin{align*}
\mathcal{N}_{X\to Y}(\rho_X)
&=
\operatorname{Tr}_Z
\left[
V_{\gamma}\rho_XV_{\gamma}^{\dagger}
\right],\\
\mathcal{N}_{X\to Z}(\rho_X)
&=
\operatorname{Tr}_Y
\left[
V_{\gamma}\rho_XV_{\gamma}^{\dagger}
\right].
\end{align*}

Let the innocent state be the maximally mixed state $\rho_0=\frac{\dsI}{2}$. Therefore, 
\begin{equation*}
\sigma_0
=
\mathcal{N}_{X\to Z}(\rho_0)
=
\begin{pmatrix}
1-\frac{\gamma}{2} & 0\\
0 & \frac{\gamma}{2}
\end{pmatrix}.
\end{equation*}
Now, consider Theorem~\ref{thm:Achievable_Quantum} with $\abs{\calU}=2$, $p_U(0)=\alpha, p_U(1)=1-\alpha$, $\theta_X^{(0)}=\ketbra{0}{0}$, and $\theta_X^{(1)}=\ketbra{1}{1}$. 
The resulting average input state is $\rho_X = \alpha\ketbra{0}{0}+
(1-\alpha)\ketbra{1}{1}$. Consequently,
\begin{align*}
\rho_Z
&=
\alpha
\mathcal{N}_{X\to Z}(\ketbra{0}{0})
+
(1-\alpha)
\mathcal{N}_{X\to Z}(\ketbra{1}{1})\\
&=
\begin{pmatrix}
\alpha+(1-\alpha)(1-\gamma) & 0\\
0 & (1-\alpha)\gamma
\end{pmatrix},
% =
% \sigma_0,
\end{align*}
and hence the covertness constraint $\rho_Z=\sigma_0$ is satisfied when $\alpha=\tfrac{1}{2}$. At the legitimate receiver, the two conditional output states are $\rho_Y^{(0)}=
\ketbra{0}{0}$, and $\rho_Y^{(1)}=
\gamma\ketbra{0}{0}+
(1-\gamma)\ketbra{1}{1}$. Therefore,
\begin{equation*}
\bbI(U;Y)=
\bbH_b\left(\frac{1-\gamma}{2}\right)
-
\frac{1}{2}\bbH_b(\gamma).
\end{equation*}
Similarly, the two conditional output states at the warden are $\rho_Z^{(0)}=
\ketbra{0}{0}$, and $\rho_Z^{(1)}=
(1-\gamma)\ketbra{0}{0}+
\gamma\ketbra{1}{1}$, which leads to
\begin{equation*}
\bbI(U;Z)=\bbH_b\left(\frac{\gamma}{2}\right)-\frac{1}{2}\bbH_b(\gamma).
\end{equation*}
Thus, by Theorem~\ref{thm:Achievable_Quantum}, the following rate pair is achievable:
\begin{align*}
R&<\bbH_b\left(\frac{1-\gamma}{2}\right)-
\frac{1}{2}\bbH_b(\gamma),\\
R_K&>\left[\bbH_b\left(\frac{\gamma}{2}\right)-\bbH_b\left(\frac{1-\gamma}{2}\right)
\right]^+.
\end{align*}
For example, let $\gamma=3/4$. In this case,
\begin{align*}
\bbI(U;Y)&=\bbH_b\left(\frac{1}{8}\right)
-\frac{1}{2}\bbH_b\left(\frac{3}{4}\right)
\approx 0.1379,\\
\bbI(U;Z)&=\bbH_b\left(\frac{3}{8}\right)
-\frac{1}{2}\bbH_b\left(\frac{3}{4}\right)
\approx 0.5488.
\end{align*}
Hence, a positive covert communication rate is achievable with $R<0.1379$, while the required secret-key rate can be chosen as $R_K>0.5488-0.1379=0.4109$.

It is worth noting that, for this example, the mixedness of the innocent state is essential. Although $\calN_{X\to Z}$ is not injective as a channel, one can show that the only input state satisfying $N_{X\to Z}(\rho_X)=\ketbra{0}{0}$ is $\rho_X=\ketbra{0}{0}$. Hence no non-trivial input state can satisfy the covertness constraint for this choice of pure innocent state. Hence, no non-trivial input ensemble can satisfy the covertness constraint, and a
positive covert communication rate is impossible.
\end{example}
\begin{lemma}[Covertness Constraint Equivalence]
\label{lemma:Covertness_Constraint_Equivalence}
Suppose that the innocent input state $\rho_0$ is pure, e.g., $\rho_0 = \ketbra{0}{0}$. Then the condition in~\eqref{eq:Covertness_Positivity} is equivalent to
\begin{align}
\exists\, \rho_X \neq \rho_0 \quad \text{such that} \quad
\Tr_Y\!\left[\calN_{X\to YZ}(\rho_X)\right] = \sigma_0.
\label{eq:Covertness_Positivity_2}
\end{align}
\end{lemma}

\begin{proof}
We first show that~\eqref{eq:Covertness_Positivity} implies~\eqref{eq:Covertness_Positivity_2}. By~\eqref{eq:Covertness_Positivity}, there exist states $\rho_1,\dots,\rho_t \in \calD(\calH_X)\setminus\{\rho_0\}$ and probabilities $p_1,\dots,p_t$ with $p_i \ge 0$ and $\sum_{i=1}^t p_i = 1$ such that
\begin{align}
\sigma_0 = \sum_{i=1}^t p_i \Tr_Y\!\left[\calN_{X\to YZ}(\rho_i)\right].
\label{eq:Equivalence_Cov}
\end{align}
Using the linearity of both the quantum channel $\calN_{X\to YZ}$ and the partial trace, we obtain
\begin{align}
\sigma_0 
= \Tr_Y\!\left[\calN_{X\to YZ}\!\left(\sum_{i=1}^t p_i \rho_i \right)\right].
\label{eq:Equivalence_Cov_2}
\end{align}
Define $\rho^\star \triangleq \sum_{i=1}^t p_i \rho_i,$
which is a valid density operator. Then~\eqref{eq:Equivalence_Cov_2} shows that
\begin{align*}
\Tr_Y\!\left[\calN_{X\to YZ}(\rho^\star)\right] = \sigma_0.
\end{align*}

It remains to show that $\rho^\star \neq \rho_0$. Suppose, for contradiction, that $\rho^\star = \rho_0$. Then $\rho_0$ is expressed as a convex combination of the states $\{\rho_i\}_{i=1}^t$, all of which are distinct from $\rho_0$. This contradicts the fact that $\rho_0$ is pure, and hence an extreme point of the convex set $\calD(\calH_X)$ \cite{Nielson_Chaung}. Therefore, $\rho^\star \neq \rho_0$, which establishes~\eqref{eq:Covertness_Positivity_2}.

The reverse implication is immediate: if there exists $\rho_X \neq \rho_0$ such that
\begin{align*}
\Tr_Y\!\left[\calN_{X\to YZ}(\rho_X)\right] = \sigma_0,
\end{align*}
then $\sigma_0$ is trivially a convex combination (with a single term) of elements in the set $\br{\Tr_Y[\calN_{X\to YZ}(\rho)]: \rho \neq \rho_0}$. This implies~\eqref{eq:Covertness_Positivity}.
\end{proof}
\begin{remark}[Injective and Non-Injective Channels]
Lemma~\ref{lemma:Covertness_Constraint_Equivalence} highlights the role of channel injectivity when the innocent input state is pure. If the warden's marginal channel $\mathcal{N}_{X\to Z}$ is injective, no input state other than the innocent state can induce the same output state at the warden. Hence, the covertness constraint $\rho_Z=\sigma_0$ admits no non-trivial solution, and no positive covert communication rate is achievable. In contrast, for a non-injective channel, another input state may induce the same output as the innocent state, allowing the covertness constraint to be satisfied and a positive covert rate to be achieved.
\end{remark}

In classical \acp{DMC}, the input alphabet is finite, and consequently the set of extreme points of the input probability simplex is finite. This finiteness enables dimension-based arguments to enforce affine dependence among the output distributions induced by the input symbols, i.e., among the points $\br{v_x}_{x\in\calX}$, where $v_x\triangleq W_{Z|X=x}$, and hence shapes the geometry of the achievable output set $\mathrm{conv}\br{v_x : x\in\calX}$. In contrast, for quantum channels, the set of admissible input states is the set of density operators, whose extreme points (pure states) form a continuous set (e.g., the Bloch sphere in the qubit case). As a result, affine dependence arguments based on counting no longer apply. Consequently, the geometric intuition and the associated results developed for classical channels in Section~\ref{sec:Classical} and for classical-quantum channels in Section~\ref{sec:Classical_Quantum_Channels} do not extend directly to the fully quantum setting.

In the following, we assume that the innocent state $\rho_0$ is a pure state and show that for point-to-point discrete memoryless quantum channels, the constraint $\rho_Z = \sigma_0$ can be satisfied for some classes of known quantum channels.
\subsubsection{Entanglement-Breaking Channels}
\label{sec:Entanglement_Breaking}
The entanglement-breaking channels are defined as \cite[Section~4.6.7]{Wilde_Book},
\begin{align}
    \calN_{X\to Z}(\rho_X)=\sum_i\tra(E_i\rho_X)\tau_i,\label{eq:Entanglement_Breaking}
\end{align}where $\br{E_i}$ is the set of \acp{POVM} and $\tau_i\in\calD(\calH)$. Let $\rho_0=\den{0}{0}$, $E_0=\tfrac{1}{2}\pr{\dsI+\sigma_x}$, and $E_1=\dsI-E_0=\tfrac{1}{2}\pr{\dsI-\sigma_x}$, where $\sigma_x$ is the Pauli-$X$ matrix. Hence,
\begin{align}
    \sigma_0=\calN_{X\to Z}(\rho_0)&=\tra(E_0\rho_0)\tau_0+\tra(E_1\rho_0)\tau_1,\nonumber\\
    &=\frac{1}{2}\tau_0+\frac{1}{2}\tau_1.\nonumber
\end{align}Now consider the input state $\rho_X=\tfrac{1}{2}\pr{\dsI+r_y\sigma_y+r_z\sigma_z}$, where $r_y,r_z\in\bbR$, such that $r_y^2+r_z^2\le1$, $(r_y,r_z)\ne(0,1)$, and $\sigma_y$ and $\sigma_z$ are the Pauli-$Y$ and Pauli-$Z$ matrices, respectively. Therefore,
\begin{align}
    \calN_{X\to Z}(\rho_X)&=\tra(E_0\rho_X)\tau_0+\tra(E_1\rho_X)\tau_1,\nonumber\\
    &=\frac{1}{2}\tau_0+\frac{1}{2}\tau_1.\nonumber
\end{align}Therefore, from Lemma~\ref{lemma:Covertness_Constraint_Equivalence} the covertness constraint $\rho_Z=\sigma_0$ can be satisfied.
\subsubsection{Partial Trace Channels}
\label{sec:Partial_Trace}
Consider a bipartite Hilbert space $\calH_X \triangleq \calH_A \otimes \calH_B$ and define the channel
\begin{align}
    \calN_{X \to Z}(\rho_{AB}) = \Tr_B(\rho_{AB}),
    \label{eq:Partial_Trace_Channel}
\end{align}
where $\Tr_B(\cdot)$ denotes the partial trace over subsystem $B$. Let the innocent input state be $\rho_0 = \dyad{0}{0}_A \otimes \dyad{0}{0}_B$, hence $\sigma_0 = \calN_{X \to Z}(\rho_0) = \dyad{0}{0}_A$. 
Now, for any state $\sigma_B \neq \dyad{0}{0}_B$, we have $\calN_{X \to Z}\!\left(\dyad{0}{0}_A \otimes \sigma_B\right)
= \dyad{0}{0}_A
= \sigma_0$. Hence, there exists a state $\rho_X \neq \rho_0$ such that $\calN_{X \to Z}(\rho_X)=\sigma_0$. Therefore, by Lemma~\ref{lemma:Covertness_Constraint_Equivalence}, the covertness constraint $\rho_Z = \sigma_0$ can be satisfied.

\subsubsection{Leakage-Reset Channel}
Let the input Hilbert space be $\mathrm{span}\br{\ket{0},\ket{1},\ket{2}}$ and define the channel~as
\begin{align}
    \calN_{X\to Z}(\rho)=\Pi\rho\Pi+\tra\sbr{(\dsI-\Pi)\rho}\den{0}{0},\label{eq:Leakage_Reset_Ch}
\end{align}where $\Pi\triangleq\den{0}{0}+\den{1}{1}$. Assuming that $\rho_0=\den{0}{0}$ we have $\calN_{X\to Z}(\rho_0)=\den{0}{0}$. We also have $\calN_{X\to Z}(\den{2}{2})=\den{0}{0}$ therefore by Lemma~\ref{lemma:Covertness_Constraint_Equivalence}, the covertness constraint $\rho_Z = \sigma_0$ can be satisfied.

\section{Conclusions}
\label{sec:Conclusions}

In this paper, we investigated the conditions under which covert communication at positive rates is achievable over classical, classical-quantum, and fully quantum point-to-point and multiple-access memoryless channels. For classical \acp{DMC}, we identified geometric conditions under which positive covert communication rates are achievable and established the corresponding covert capacity results. We showed that when the input alphabet has sufficiently large cardinality relative to the output alphabet, the innocent output distribution can often be expressed as a convex combination of the output distributions induced by communication symbols, thereby enabling positive covert communication rates. As a natural approach to enlarge the effective channel input alphabet, we study conditions under which the covertness constraint can be satisfied for \acp{MAC}. We show that the presence of an additional transmitter can facilitate satisfying the covertness constraint, thereby enabling positive covert rates in settings where the corresponding point-to-point channel may not permit them. We further extended this geometric perspective to classical-quantum channels by characterizing the role of the affine dimension of the induced quantum states and established positive covert rate results in the quantum setting. Finally, we show that, for fully quantum channels, the covertness constraint can be satisfied with a non-trivial input state whenever the innocent input state is~mixed.

\begin{appendices}

\section{Proof of Theorem~\ref{thm:Achievable_Quantum}}
\label{proof:thm:Achievable_Quantum}
Fix the distributions $p_U$, and a collection of channel input states $\br{\theta_X^{(u)}}_{u\in\calU}$, and let $\rho_0$ be the innocent input state and $\sigma_0\triangleq\tra_Y\sbr{\calN_{X\to YZ}(\rho_0)}$.
\subsection{Codebook Generation}Let $C_n\triangleq\big\{U^n(k,m)\big\}_{(k,m)\in\calM\times\calK}$, where $\calK\triangleq\sbra{1}{2^{nR_K}}$ and $\calM\triangleq\sbra{1}{2^{nR}}$, be a random codebook generated \ac{iid} according~to $p_U$ and $\calC_n\triangleq\big\{u^n(k,m)\big\}_{(k,m)\in\calM\times\calK}$ be a realization of the codebook $C_n$.

\subsection{Encoding and Decoding}Given the secret key $k\in\calK$ and the message $m\in\calM$, the encoder computes the codeword $U^n(k,m)$ and then prepares $\bigotimes_{t=1}^n\theta_X^{(u_t(k,m))}$ and transmits $\theta_X^{(u_t(k,m))}$ over~$\calN_{X\to YZ}$, for $t\in\sbra{1}{n}$. 
Then, for any $\epsilon > 0$, by the Packing~Lemma~{\cite[Lemma~16.3.1]{Wilde_Book}}, for each $k\in\calK$, there exists a set of \ac{POVM} operators $\{\Lambda_m^{(k)}\}_{m \in \calM_n}$ and a constant $c > 0$ such that, averaged over the random codebook ensemble, the probability of error is upper bounded by 
\begin{align}
 1-\bbE_{C_n}\left[\frac{1}{\card{\calM_n}}\sum_{m\in\calM_n}\tra\left(\Lambda_m^{(k)}\rho_{Y^n}^{U^n(k,m)}\right)\right]\le2(\epsilon+2\sqrt{\epsilon})+4\times2^{n(R-I(U;Y)+2c\epsilon)},\label{eq:PError}
\end{align}where the inequality holds by \cite[Lemma~16.3.1]{Wilde_Book} and the \ac{RHS} of \eqref{eq:PError} vanishes as $n$ grows~when
\begin{align}
    R<I(U;Y). \label{eq:Decodability}
\end{align}

\subsection{Covertness Analysis}
Our covertness analysis is based on the following lemma from \cite[Lemma~3]{Quantum_Covert_CSI}.
\begin{lemma}
    \label{lemma:Res_Quantum}
    Let $\rho_{UZ}\triangleq\sum_{u\in\calU}p_U(u)\den{u}{u}\otimes\rho_{Z\lvert u}$ be a classical-quantum state. Also, let $C_n\triangleq\{U^n(m)\}_{m\in\calM_n}$, where $\calM_n\triangleq\left[2^{nR}\right]$, be a set of random variables in which $U^n(m)$ is generated \ac{iid} according to $p_U$. Then, for $R>I(U;Z)$,
    \begin{align}
        \bbE_{C_n}\bbD\left(\frac{1}{2^{nR}}\sum_{m}\rho_{Z^n\lvert U^n(m)}\Big\lVert\rho_Z^{\otimes n}\right)\xrightarrow[n\to\infty]{}0.\nonumber
    \end{align}
\end{lemma}
\begin{proof}
    For completeness, the proof is provided in Appendix~\ref{proof:lemma:Res_Quantum}.
\end{proof}
To prove the covertness of our code design, we bound $\bbE_{C_n}\bbD\left(\hat{\rho}_{Z^n}\lVert\sigma_Z^{\otimes n}\right)$, where $\hat{\rho}_{Z^n}$ is the distribution induced by our code design and is $\hat{\rho}_{Z^n}=\frac{1}{2^{n(R_K+R)}}\sum_{k\in\calK}\sum_{m\in\calM}\rho_{Z^n\lvert U^n(k,m)}$, and then choose $p_U$ and $\br{\theta_X^{(u)}}_{u\in\calU}$ such that $\rho_Z=\sigma_0$. Now from Lemma~\ref{lemma:Res_Quantum} and \cite[Theorem~11.9.1]{Wilde_Book}, we have
\begin{align}
\bbE_{C_n}\left\lVert\hat{\rho}_{Z^n}-\sigma_Z^{\otimes n}\right\rVert_1\xrightarrow[n\to\infty]{}0,\nonumber
\end{align}if $R_K+R>I(U;Z)$, which with  \eqref{eq:Decodability} complete the proof of Theorem~\ref{thm:Achievable_Quantum}.

\section{Proof of Lemma~\ref{lemma:Res_Quantum}}
\label{proof:lemma:Res_Quantum}
We have, 
\begin{align}
    &\bbE_{C_n}\bbD\left(\frac{1}{2^{nR}}\sum_{m}\rho_{Z^n\lvert U^n(m)}\Big\lVert\rho_Z^{\otimes n}\right)\nonumber \displaybreak[0]\\
    &=\bbE_{C_n}\tra\left(\frac{1}{2^{nR}}\sum_m\rho_{Z^n\lvert U^n(m)}\left[\log\left(\frac{1}{2^{nR}}\sum_{m'}\rho_{Z^n\lvert U^n(m')}\right)-\log\left(\rho_Z^{\otimes n}\right)\right]\right)\nonumber\\
    &=\frac{1}{2^{nR}}\tra\left(\sum_m\bbE_{C_n}\rho_{Z^n\lvert U^n(m)}\left[\log\left(\frac{1}{2^{nR}}\rho_{Z^n\lvert U^n(m)}+\frac{1}{2^{nR}}\sum_{m'\ne m}\rho_{Z^n\lvert U^n(m')}\right)-\log\left(\rho_Z^{\otimes n}\right)\right]\right)\nonumber\\
    &\mathop\le\limits^{(a)}\frac{1}{2^{nR}}\tra\left(\sum_m\bbE_{U^n(m)}\rho_{Z^n\lvert U^n(m)}\left[\log\left(\frac{1}{2^{nR}}\rho_{Z^n\lvert U^n(m)}+\frac{1}{2^{nR}}\sum_{m'\ne m}\bbE_{U^n(m')}\rho_{Z^n\lvert U^n(m')}\right)-\log\left(\rho_Z^{\otimes n}\right)\right]\right)\nonumber\\
    &\frac{1}{2^{nR}}\tra\left(\sum_m\bbE_{U^n(m)}\rho_{Z^n\lvert U^n(m)}\left[\log\left(\frac{1}{2^{nR}}\rho_{Z^n\lvert U^n(m)}+\frac{\left(2^{nR}-1\right)}{2^{nR}}\rho_Z^{\otimes n}\right)-\log\left(\rho_Z^{\otimes n}\right)\right]\right)\nonumber\\
    &\mathop\le\limits^{(b)}\frac{1}{2^{nR}}\tra\left(\sum_m\bbE_{U^n(m)}\rho_{Z^n\lvert U^n(m)}\left[\log\left(\frac{v_Z}{2^{nR}}\calE_{\rho_Z^\on}\pr{\rho_{Z^n\lvert U^n(m)}}+\frac{\left(2^{nR}-1\right)}{2^{nR}}\rho_Z^{\otimes n}\right)-\log\left(\rho_Z^{\otimes n}\right)\right]\right)\nonumber\\
    &\mathop\le\limits^{(c)}\frac{1}{2^{nR}}\tra\left(\sum_m\bbE_{U^n(m)}\rho_{Z^n\lvert U^n(m)}\log\left(I+\frac{v_Z}{2^{nR}}\calE_{\rho_Z^{\otimes n}}\pr{\rho_{Z^n\lvert U^n(m)}}\left(\rho_Z^{\otimes n}\right)^{-1}\right)\right)\nonumber\\
    &=\frac{1}{\alpha2^{nR}}\tra\left(\sum_m\bbE_{U^n(m)}\rho_{Z^n\lvert U^n(m)}\log\left(I+\frac{v_Z}{2^{nR}}\calE_{\rho_Z^{\otimes n}}\pr{\rho_{Z^n\lvert U^n(m)}}\left(\rho_Z^{\otimes n}\right)^{-1}\right)^\alpha\right)\nonumber\\
    &\mathop\le\limits^{(d)}\frac{1}{\alpha2^{nR}}\tra\left(\sum_m\bbE_{U^n(m)}\rho_{Z^n\lvert U^n(m)}\frac{v_Z^\alpha}{2^{\alpha nR}}\calE_{\rho_Z^{\otimes n}}\pr{\rho_{Z^n\lvert U^n(m)}}^\alpha\left(\rho_Z^{\otimes n}\right)^{-\alpha}\right)\nonumber\\
    &\mathop=\limits^{(e)}\frac{v_Z^\alpha}{\alpha2^{n\alpha R}}\tra\left(\bbE_{U^n}\rho_{Z^n\lvert U^n}\calE_{\rho_Z^{\otimes n}}\pr{\rho_{Z^n\lvert U^n}}^{\alpha}\left(\rho_Z^{\otimes n}\right)^{-\alpha}\right)\nonumber\\
    &=\frac{v_Z^\alpha}{\alpha2^{n\alpha R}}\tra\left(\bbE_{U^n}\rho_{Z^n\lvert U^n}\calE_{\rho_Z^{\otimes n}}\pr{\calE_{\rho_Z^{\otimes n}}\pr{\rho_{Z^n\lvert U^n}}^{\alpha}\left(\rho_Z^{\otimes n}\right)^{-\alpha}}\right)\nonumber\\
    &\mathop=\limits^{(f)}\frac{v_Z^\alpha}{\alpha2^{n\alpha R}}\tra\left(\bbE_{U^n}\calE_{\rho_Z^{\otimes n}}\pr{\rho_{Z^n\lvert U^n}}^{1+\alpha}\left(\rho_Z^{\otimes n}\right)^{-\alpha}\right)\nonumber\\
    &\mathop=\limits^{(g)}\frac{v_Z^\alpha}{\alpha2^{n\alpha R}}\tra\left(\calE_{\rho_{U^n}\otimes\rho_Z^{\otimes n}}\pr{\rho_{U^nZ^n}}^{1+\alpha}\left(\rho_{U^n}\otimes\rho_Z^{\otimes n}\right)^{-\alpha}\right)\nonumber\\
    &\mathop\le\limits^{(h)}\frac{v_Z^\alpha}{\alpha}2^{\alpha(-nR+\ubar{\D}_{1+\alpha}(\rho_{U^nZ^n};\rho_{U^n}\otimes \rho_Z^{\otimes n}))}\nonumber\\
    &\mathop\le\limits^{(i)}\frac{v_Z^\alpha}{\alpha}2^{\alpha n(-R+\ubar{\D}_{1+\alpha}(\rho_{UZ};\rho_{U}\otimes \rho_{Z}))},\label{eq:Resolv_Analysis}
\end{align}where
\begin{itemize}
    \item[$(a)$] follows from the linearity of the expectation, the law of total expectations, Jensen's inequality, and since the symbols $U(i)$ and $U(i')$, for $i\ne i'$, are independent;
    \item[$(b)$] follows since for two quantum states $\sigma\in\calD(\calH)$ and $\rho\in\calD(\calH)$ we have \cite{Hayashi02},
    \begin{align}
        \rho\preceq v\calE_\sigma(\rho),\label{eq:Pinching_Inequlaity}
    \end{align}where $v$ is the distinct number of eigenvalues of $\sigma$;
    \item[$(c)$] follows since $\log$ is a matrix monotone function;
    \item[$(d)$] follows since $\calE_{\rho_E^{\otimes n}}\pr{\rho_{E^n\lvert U^n(m)}}$ commutes with $\rho_E^{\otimes n}$ and for $a,b\in\bbR^+$, $0<\alpha\le1$, we have $\frac{1}{\alpha}\log(1+x)^\alpha\le\frac{1}{\alpha}\log(1^\alpha+x^\alpha)\le\frac{x^\alpha}{\alpha}$;
    \item[$(e)$] follows from the symmetry of the codebook construction \ac{wrt} the message;
    \item[$(f)$] follows since for three arbitrary states $\sigma\in\calD(\calH)$, with spectral decomposition $\sigma=\sum_i\lambda_i\den{x_i}{x_i}$, and $\rho_1\in\calD(\calH)$, and $\rho_2\in\calD(\calH)$ we have
    \begin{align}
    \tra\left[\calE_\sigma(\rho_1)\rho_2\right]&=\sum_i\bra{x_i}\rho_1\ket{x_i}\tra\left[\den{x_i}{x_i}\rho_2\right]\nonumber\\
    &=\sum_i\bra{x_i}\rho_1\ket{x_i}\tra\left[\bra{x_i}\rho_2\ket{x_i}\right]\nonumber\\
    &=\sum_i\tra\left[\bra{x_i}\rho_1\ket{x_i}\right]\bra{x_i}\rho_2\ket{x_i}\nonumber\\
    &=\sum_i\tra\left[\rho_1\den{x_i}{x_i}\right]\bra{x_i}\rho_2\ket{x_i}\nonumber\\
    &=\sum_i\tra\left[\rho_1\bra{x_i}\rho_2\ket{x_i}\den{x_i}{x_i}\right]\nonumber\\
    &=\tra[\rho_1\calE_\sigma(\rho_2)];\nonumber
    \end{align}
    \item[$(g)$] follows since the states are classical-quantum;
    \item[$(h)$] follows from the definition of the sandwiched R\'{e}nyi relative entropy and since for two  states $\rho,\sigma\in\calD(\calH)$ and $\calE(\cdot):\calL(A)\to\calL(B)$, we have $\ubar{\D}_{1+\alpha}\pr{\calE(\rho);\calE(\sigma)}\le\ubar{\D}_{1+\alpha}\pr{\rho;\sigma}$;
    \item[$(i)$] follows from the fact that by the codebook construction, $U^n$ is an \ac{iid} sequence.
\end{itemize}Therefore, when $n\to\infty$ and $\alpha\to0$ the \ac{RHS} of \eqref{eq:Resolv_Analysis} vanishes if $ R>I(U;Z)$.

\section{Converse Proof of Theorem~\ref{thm:Capacity_Classical_Sto}}
\label{proof:thm:capacity}
Consider any sequence of $\pr{2^{n(R+R_K)},n}$ stochastic codes for a \ac{DMC} $W_{YZ|X}$, that simultaneously satisfies the reliability constraint $\bbP\br{M\ne\hat{M}}\le\gamma_n$, where $\lim\limits_{n\to\infty}\gamma_n=0$, and the covertness constraint $\left\lVert p_{Z^n}-q_0^{\otimes n}\right\lVert_1\le\delta_n$, where $\lim\limits_{n\to\infty}\delta_n=0$.

\textit{$\epsilon$-Rate Region:} We first define a $\calQ_\epsilon$ region, which expands the region in \eqref{eq:Capacity_Classical_All_Sto}, as follows
\begin{align}
  \calQ_\epsilon= \left.\begin{cases}R\geq 0,R_K\ge0: \exists\, p_{UXYZ}\in\calG_\epsilon:\\
  R\le\bbI(U;Y)+\epsilon\\
  R_K\ge\bbI(U;Z)-\bbI(U;Y)-3\epsilon\\
\end{cases}\right\},\nonumber
\end{align}
where,
\begin{align}
  \calG_\epsilon \triangleq \left.\begin{cases}p_{UXYZ}:\\
p_{UXYZ}=p_{UX}\pr{u,x}W_{YZ|X}\\
\left\lVert p_Z-q_0\right\lVert_1\le\epsilon\\
\end{cases}\right\}.\nonumber
\end{align}
Fix $\delta>0$. We show that any achievable rate $R$ must belong to $\calQ_\delta$. For arbitrary $\gamma_n,\delta_n>0$, the rate can be upper-bounded as follows:
\begin{subequations}
\begin{align}
    nR&=\bbH(M)\nonumber\\
    &\mathop=\limits^{(a)}\bbH(M|K)\nonumber\\
    &\mathop\le\limits^{(b)}\bbI(M;Y^n|K)+n\gamma_n\label{eq:Rel_Converse_1}\\
    &=\sum_{t=1}^n\bbI(M;Y_t|K,Y^{t-1})+n\gamma_n\nonumber\\
    &\le\sum_{t=1}^n\bbI(M,K,Y^{t-1};Y_t)+n\gamma_n\label{eq:Rel_Converse_2}\\
    &\le\sum_{t=1}^n\bbI(M,K,Y^{t-1},Z_{t+1}^n;Y_t)+n\gamma_n\nonumber\\
    &\mathop=\limits^{(c)}\sum_{t=1}^n\bbI(U_t;Y_t)+n\gamma_n\nonumber\\
    &\mathop=\limits^{(d)}n\sum_{t=1}^n\bbP(T=t)\bbI(U_t;Y_t|T=t)+n\gamma_n\nonumber\\
    &=n\bbI(U_T;Y_T|T)+n\gamma_n\nonumber\\
    &\le n\bbI(U_T,T;Y_T)+n\gamma_n\nonumber\\
    &\mathop=\limits^{(e)} n\bbI(U;Y)+n\gamma_n\nonumber\\
    &\mathop=\limits^{(f)} n\bbI(U;Y)+n\epsilon,\label{eq:Rel_Converse_3}
\end{align}where
\begin{itemize}
    \item[$(a)$] follows since $M$ and $K$ are independent;
    \item[$(b)$] follows from Fano's inequality;
    \item[$(c)$] follows by defining
    \begin{align}
        U_t\triangleq\pr{M,K,Y^{t-1},Z_{t+1}^n};\label{eq:Defi_Ut}
    \end{align}
     \item[$(d)$] follows by introducing a random variable $T$, uniformly distributed over $[1\!:\!n]$, that is independent of all other random variables;
     \item[$(e)$] follows by defining
    \begin{align}
        U\triangleq\pr{U_T,T},\quad Y\triangleq Y_T,\quad Z\triangleq Z_T;\label{eq:Defi_UYZ}
    \end{align}
    \item[$(f)$] follows by defining 
    \begin{align}
    \epsilon\triangleq\max\br{\gamma_n,\delta'_n,\delta_n},\label{eq:Defi_epsil}
    \end{align}where $\delta'_n\triangleq4\delta_n\pr{\log\pr{\abs{\calZ}}+\log\left(\frac{1}{\delta_n}\right)}$.
\end{itemize}
\end{subequations}
We can also lower bound the rate as follows,
\begin{align}
    n(R+R_K)&=\bbH(M,K)\nonumber\\
    &\ge\bbI(M,K;Z^n)\nonumber\\
    &=\sum_{t=1}^n\bbI\pr{M,K;Z_t|Z_{t+1}^n}\nonumber\\
    &\ge\sum_{t=1}^n\bbI\pr{M,K,Z_{t+1}^n;Z_t}-n\delta'_n,\label{eq:Res_Converse_1}
\end{align}where the last inequality follows from \cite[Lemma~VI.3]{Cuff13} noting that $\delta'_n\triangleq4\delta_n\pr{\log\pr{\abs{\calZ}}+\log\left(\frac{1}{\delta_n}\right)}$. Now combining \eqref{eq:Rel_Converse_2} with \eqref{eq:Res_Converse_1} leads to
\begin{align}
    R_K&\ge\frac{1}{n}\sum_{t=1}^n\big[\bbI\pr{M,K,Z_{t+1}^n;Z_t}-\bbI(M,K,Y^{t-1};Y_t)\big]-\gamma_n-\delta'_n\nonumber\\
    &=\frac{1}{n}\sum_{t=1}^n\big[\bbI\pr{M,K,Y^{t-1},Z_{t+1}^n;Z_t}-\bbI\pr{M,K,Y^{t-1},Z_{t+1}^n;Y_t}\nonumber\\
    &\quad-\bbI\pr{Y^{t-1};Z_t|M,K,Z_{t+1}^n}+\bbI\pr{Z_{t+1}^n;Y_t|M,K,Y^{t-1}}\big]-\gamma_n-\delta'_n\nonumber\\
    &\mathop=\limits^{(a)}\frac{1}{n}\sum_{t=1}^n\big[\bbI\pr{M,K,Y^{t-1},Z_{t+1}^n;Z_t}-\bbI\pr{M,K,Y^{t-1},Z_{t+1}^n;Y_t}\big]-\gamma_n-\delta'_n\nonumber\\
    &\mathop=\limits^{(b)}\frac{1}{n}\sum_{t=1}^n\big[\bbI\pr{U_t;Z_t}-\bbI\pr{U_t;Y_t}\big]-\gamma_n-\delta'_n\nonumber\\
    &\mathop=\limits^{(c)}\sum_{t=1}^n\bbP(T=t)\big[\bbI\pr{U_t;Z_t|T=t}-\bbI\pr{U_t;Y_t|T=t}\big]-\gamma_n-\delta'_n\nonumber\\
    &=\bbI\pr{U_T;Z_T|T}-\pr{U_T;Y_T|T}-\gamma_n-\delta'_n\nonumber\\
    &\ge\bbI\pr{U_T;Z_T|T}-\bbI\pr{U_T,T;Y_T}-\gamma_n-\delta'_n\nonumber\\
    &\mathop\ge\limits^{(d)}\bbI\pr{U_T,T;Z_T}-\bbI\pr{U_T,T;Y_T}-\gamma_n-2\delta'_n\nonumber\\
    &\mathop=\limits^{(e)}\bbI\pr{U;Z}-\bbI\pr{U;Y}-\gamma_n-2\delta'_n\nonumber\\
    &\mathop\ge\limits^{(f)}\bbI\pr{U;Z}-\bbI\pr{U;Y}-3\epsilon,\label{eq:Res_Converse_2}
\end{align}where
\begin{itemize}
    \item[$(a)$] follows from K\"orner-Marton-Csisz\'ar sum identity \cite{BCC:IT78};
    \item[$(b)$] follows from \eqref{eq:Defi_Ut};
    \item[$(c)$] follows by introducing a random variable $T$, uniformly distributed over $[1\!:\!n]$, that is independent of all other random variables;
    \item[$(d)$] follows from \cite[Lemma~VI.3]{Cuff13};
    \item[$(e)$] follows from \eqref{eq:Defi_UYZ};
    \item[$(f)$] follows from \eqref{eq:Defi_epsil}. 
\end{itemize}
Now, we have 
\begin{align}
\left\lVert p_Z-q_0\right\rVert_1&=\left\lVert p_{Z_T}-q_0\right\rVert_1\nonumber\\
&=\left\lVert\frac{1}{n}\sum_{t=1}^np_{Z_t}-q_0\right\rVert_1\nonumber\\
&\mathop\leq\limits^{(a)}\frac{1}{n}\sum_{t=1}^n\left\lVert p_{Z_t}-q_0\right\rVert_1\nonumber\\
&\mathop\leq\limits^{(b)}\left\lVert p_{Z^n}-q_0^\on\right\rVert_1\nonumber\\
&\le\delta_n,\label{eq:Covertness_Converse}
\end{align}where $(a)$ follows from the triangle inequality, and $(b)$ follows from the monotonicity of the total variation distance.

The proof of continuity at zero of $\calQ_\epsilon$ is similar to that of \cite[Appendix~F]{Keyless22} and is omitted.

\end{appendices}
\bibliographystyle{IEEEtran}
\bibliography{IEEEabrv,bibfile}

\end{document}

%% file: CommandsAndMacros.tex
\newcommand{\dsI}{\mathds{I}}

\newcommand{\calB}{\mathcal{B}}

\newcommand{\calC}{\mathcal{C}}

\newcommand{\calD}{\mathcal{D}}
\newcommand{\bbD}{\mathbb{D}}

\newcommand{\calE}{\mathcal{E}}
\newcommand{\bbE}{\mathbb{E}}

\newcommand{\calF}{\mathcal{F}}

\newcommand{\calG}{\mathcal{G}}

\newcommand{\calH}{\mathcal{H}}
\newcommand{\bbH}{\mathbb{H}}

\newcommand{\bbI}{\mathbb{I}}

\newcommand{\calK}{\mathcal{K}}

\newcommand{\calL}{\mathcal{L}}

\newcommand{\calM}{\mathcal{M}}

\newcommand{\calN}{\mathcal{N}}

\newcommand{\calP}{\mathcal{P}}
\newcommand{\bbP}{\mathbb{P}}

\newcommand{\calQ}{\mathcal{Q}}

\newcommand{\bbR}{\mathbb{R}}

\newcommand{\calS}{\mathcal{S}}

\newcommand{\calU}{\mathcal{U}}

\newcommand{\calX}{\mathcal{X}}

\newcommand{\calY}{\mathcal{Y}}

\newcommand{\calZ}{\mathcal{Z}}

\newcommand\ubar[1]{\stackunder[1.1pt]{$#1$}{\rule{1.2ex}{.08ex}}}

\def\on{{\otimes n}}

\DeclareMathOperator{\tra}{Tr}
\DeclareMathOperator{\D}{D}

\newcommand{\den}[2]{\ensuremath{\ket{#1}{\hspace{-1.6mm}}\bra{#2}}}

\newtheorem{theorem}{Theorem}
\newtheorem{corollary}{Corollary}
\newtheorem{definition}{Definition}

\newtheorem{remark}{Remark}

\newtheorem{example}{Example}

\newtheorem{lemma}{Lemma}[]

\newcommand{\indic}[1]{\ensuremath{\mathds{1}}}
\newcommand{\card}[1]{\ensuremath{\left\lvert{#1}\right\rvert}}   % Absolute value \newcommand{\expec}[1]{\ensuremath{\mathds{E}}}

\newcommand{\sbra}[2]{\ensuremath{\left[{#1}{\,:\,}{#2}\right]}}%
\newcommand{\sbr}[1]{\ensuremath{\left[{#1}\right]}}%
\newcommand{\pr}[1]{\ensuremath{\left({#1}\right)}}%
\newcommand{\br}[1]{\ensuremath{\left\{{#1}\right\}}}%

%% file: Acronyms.tex
\acrodef{ACDIS}[ACDIS]{Adaptive Communication Decision and Information Systems}
\acrodef{AEP}{Asymptotic Equipartition Property}
\acrodef{AoA}{Angle of Arrival}
\acrodef{AWGN}{Additive White Gaussian Noise}
\acrodef{AVC}[AVC]{Arbitrarily Varying Channel}
\acrodefplural{AVC}{Arbitrarily Varying Channels}
\acrodef{PIR-PNSI}{Private Information Retrieval with Private Noisy Side Information}
\acrodef{BER}{Bit-Error-Rate}
\acrodef{BEC}{Binary Erasure Channel}
\acrodefplural{BEC}{Binary Erasure Channels}
\acrodef{BSC}{Binary Symmetric Channel}
\acrodefplural{BSC}{Binary Symmetric Channels}
\acrodef{BSCO}{Binary Symmetric Channel with Additional ``off'' Symbol}
\acrodefplural{BSCO}{Binary Symmetric Channels with Additional ``off'' Symbols}
\acrodef{FDG}{Functional Dependence Graph}
\acrodef{BPSK}{Binary Phase-Shift Keying}
\acrodef{BICM}[BICM]{Bit-Interleaved Coded-Modulation}
\acrodef{CDF}[CDF]{cumulative distribution function}
\acrodef{CGF}[CGF]{Cumulant Generating Function}
\acrodef{CLT}[CLT]{Central Limit Theorem}
\acrodef{CSI}[CSI]{Channel State Information}
\acrodef{DMC}[DMC]{Discrete Memoryless Channel}
\acrodefplural{DMC}{Discrete Memoryless Channels}
\acrodef{DMQC}[DMQC]{Discrete Memoryless Quantum Channel}
\acrodefplural{DMQC}{Discrete Memoryless Quantum Channels}
\acrodef{DMCQC}[DMCQC]{Discrete Memoryless Classical-Quantum Channel}
\acrodefplural{DMCQC}{Discrete Memoryless Classical-Quantum Channels}
\acrodef{DMS}[DMS]{Discrete Memoryless Source}
\acrodef{ERM}[ERM]{Empirical Risk Minimization}
\acrodef{FER}[FER]{Frame Error Rate}
\acrodef{ICA}[ICA]{Independent Component Analysis}
\acrodef{iid}[i.i.d.]{independent and identically distributed}
\acrodef{IoT}[IoT]{Internet of Things}
\acrodef{KKT}[KKT]{Karush-Kuhn Tucker}
\acrodef{LASSO}[LASSO]{Least Absolute Shrinkage and Selection Operator}
\acrodef{LPD}[LPD]{Low Probability of Detection}
\acrodef{LDPC}[LDPC]{Low-Density Parity-Check}
\acrodef{LLMS}[LLMS]{Linear Least Mean Square}
\acrodef{LMS}[LMS]{Least Mean Square}
\acrodef{MAC}[MAC]{Multiple-Access Channel}
\acrodef{ADSI}[ADSI]{Action-Dependent State Information}
\acrodef{MGF}[MGF]{Moment Generating Function}
\acrodef{MLC}[MLC]{Multi-Level Coding}
\acrodef{MLE}[MLE]{Maximum Likelihood Estimate}
\acrodef{MIMO}[MIMO]{Multiple-Input Multiple-Output}
\acrodef{MISO}{Multiple-Input Single-Output}
\acrodef{MSD}[MSD]{Multi-Stage Decoding}
\acrodef{MMSE}[MMSE]{Minimum Mean-Square Error}
\acrodef{PAC}[PAC]{Probably Approximately Correct}
\acrodef{PCA}[PCA]{Principal Component Analysis}
\acrodef{PDF}[PDF]{probability density function}
\acrodefplural{PDF}{probability density functions}
\acrodef{PMF}[PMF]{Probability Mass Function}
\acrodefplural{PMF}{Probability Mass Functions}
\acrodef{PPM}[PPM]{Pulse Position Modulation}
\acrodef{PSD}{Power Spectral Density}
\acrodef{PSK}{Phase Shift Keying}
\acrodef{QKD}{Quantum Key Distribution}
\acrodef{ROC}{Receiver Operating Characteristic}
\acrodef{CVQKD}{Continuous-Variable \ac{QKD}}
\acrodef{QPSK}{Quadrature Phase-Shift Keying}
\acrodef{RV}{random variable}
\acrodefplural{RV}{random variables}
\acrodef{SIMO}{Single-Input Multiple-Output}
\acrodef{SNR}{Signal-to-Noise Ratio}
\acrodef{SVM}[SVM]{Support Vector Machine}
\acrodef{TPCP}{Trace-Preserving Completely-Positive}
\acrodef{wrt}[w.r.t.]{with respect to}
\acrodef{WSS}{Wide Sense Stationary}
\acrodef{RHS}{right hand side}
\acrodef{LHS}{left hand side}
\acrodef{PIR}{Private Information Retrieval}
\acrodef{MDS}{Maximum Distance Separable}
\acrodef{LLN}{Law of Large Numbers}
\acrodef{DFRC}{Dual-Function Radar Communication}
\acrodef{ISAC}{Integrated Sensing and Communication}
\acrodef{RadCom}{Joint Radar and Communicatins}
\acrodef{POVM}{Positive Operator-Valued Measure}
\acrodefplural{POVM}{Positive Operator-Valued Measures}